\documentclass[runningheads]{llncs}
\usepackage{amsmath,amssymb,graphicx,xcolor,pgfplots,enumitem,comment,csquotes,booktabs,multirow}
\pgfplotsset{compat=newest} 
\usepackage{authblk,lipsum}
\usepackage{wrapfig}
\usepackage{cancel,hyperref,cleveref}

\newcommand{\R}{\mathbb R}

\renewcommand{\l}{\left(}
\renewcommand{\r}{\right)}
\newcommand{\lt}{\left[}
\newcommand{\rt}{\right]}

\newcommand{\ol}{\overline}

\DeclareMathOperator{\relu}{ReLu}

\newcommand{\tab}{\hspace{0.3cm}}

\newcommand{\bq}{\textasciigrave}
\newcommand{\querytxta}[1]{\textcolor{blue!80!green}{#1}}
\newcommand{\querytxtb}[1]{\textcolor{red!80!green}{#1}}
\newcommand{\querytxtc}[1]{\textcolor{green!50!black}{#1}}
\newcommand{\props}{PBFURs}
\newcommand{\prop}{PBFUR}

\Crefname{notat}{Notation}{Notations}

\Crefname{observ}{Observation}{Observations}

\title{Optimality Despite Chaos in Fee Markets}

\newcommand{\hsp}{\hspace{0.05cm}}
\author{\mbox{Stefanos Leonardos$^{1*}$ \hsp Dani\"el Reijsbergen$^{2*}$ \hsp Barnab\'e Monnot$^3$ \hsp Georgios Piliouras$^2$}}
\institute{
$^1$King's College, London, UK\\$^2$ Singapore University of Technology and Design, Singapore \\$^3$ Ethereum Foundation, Berlin, Germany \\[0.2cm] $^*$The first and second authors contributed equally.}
\authorrunning{Leonardos, Reijsbergen, Monnot, and Piliouras.}

\begin{document}

\maketitle

\begin{abstract}

Transaction fee markets are essential components of block\-chain economies, as they resolve the inherent scarcity in the number of transactions that can be added to each block. In early blockchain protocols, this scarcity was resolved through a first-price auction in which users were forced to guess appropriate bids from recent blockchain data. Ethereum's EIP-1559 fee market reform streamlines this process through the use of a base fee that is increased (or decreased) whenever a block exceeds (or fails to meet) a specified target block size. Previous work has found that the EIP-1559 mechanism may lead to a base fee process that is inherently chaotic, in which case the base fee does \textit{not} converge to a fixed point even under ideal conditions. However, the impact of this chaotic behavior on the fee market's main design goal -- blocks whose long-term average size equals the target -- has not previously been explored. As our main contribution, we derive near-optimal upper and lower bounds for the time-average block size in the EIP-1559 mechanism despite its possibly chaotic evolution.
Our lower bound is equal to the target utilization level whereas our upper bound is $\approx\!6\%$ higher than optimal. Empirical evidence is shown in great agreement with these theoretical predictions. Specifically, the historical average was $\approx\!2.9\%$ larger than the target rage under Proof-of-Work and decreased to $\approx\!2.0\%$ after Ethereum's transition to Proof-of-Stake. We also find that an approximate version of EIP-1559 achieves optimality even in the absence of convergence.

\end{abstract}

\section{Introduction}\label{sec:intro}


In the seminal Bitcoin whitepaper \cite{nakamoto2008bitcoin}, the concept of a \textit{blockchain} was introduced as a secure data structure maintained by \textit{nodes} in a peer-to-peer network.
A blockchain consists of elementary database operations called \textit{transactions} that modify a global state -- e.g., cryptocurrency ownership or the state of smart contracts.
Nodes provide an essential service to the blockchain's users by broadcasting their transactions and responding to queries about the global state \cite{gencer2018decentralization}. As such, a large and diverse network of nodes enhances \textit{decentralization} in the sense that the availability and integrity of blockchain-enabled services do not depend on a handful of entities.
As nodes execute every new transaction to maintain their view of the latest global state, the \textit{computational cost} of new transactions must be limited to avoid excluding all but the most powerful nodes.
In Ethereum \cite{buterin2014next}, this computational cost is measured through the notion of \textit{gas}, and each block has a gas \textit{limit} that is decided by the nodes.
In Ethereum's original design, each transaction has a \textit{gas price} that indicates how much its creator is willing to pay for its inclusion on the blockchain. This mechanism behaves like a \textit{first-price auction}, and shares all of its drawbacks \cite{Nis07}: users tend to bid untruthfully relative to the true valuation of their transaction, which leads to guesswork and overbidding that is detrimental to the user experience.


\textit{Ethereum Improvement Proposal (EIP) 1559} \cite{Con19} simplifies Ethereum's fee market through a protocol-set \textit{base fee}. Instead of aiming to fill each block to the limit, it aims to achieve a long-term average \textit{target}, which is half the maximum size of each block.
The base fee is automatically updated to reflect market conditions: if a block is larger than the target, then demand for transaction inclusion is too high at the current price so the base fee is increased (and vice versa for smaller blocks).
The base fee hence aims to reflect the constantly-changing \textit{market-clearing price}, which is the theoretical price at which demand for transactions is precisely such that block sizes equal the target.
To add a transaction to the blockchain, users pay the base fee per unit of spent gas -- this payment is permanently destroyed or \textit{burned} \cite{karantias2020proof}.
This mechanism is provably incentive-compatible in the sense that users bid close to their true valuation unless demand is extremely high \cite{Rou20,roughgarden2021transaction}. 
However, whether the protocol is \textit{optimal} in the sense that it achieves its main design goal -- a long-term average block size that equals the target -- has not previously been explored.
Previous work has found that the base fee may \textit{not} converge to the market-clearing price \cite{Leo21} even when market conditions remain unchanged, as the base fee process exhibits (Li-Yorke) \textit{chaos} in a wide range of market conditions. As the base fees need not converge to the market-clearing price, it is natural to ask whether the long-term average block sizes in fact converge to the target.

\begin{figure}[t]
    \centering
    \includegraphics[width=0.8\linewidth]{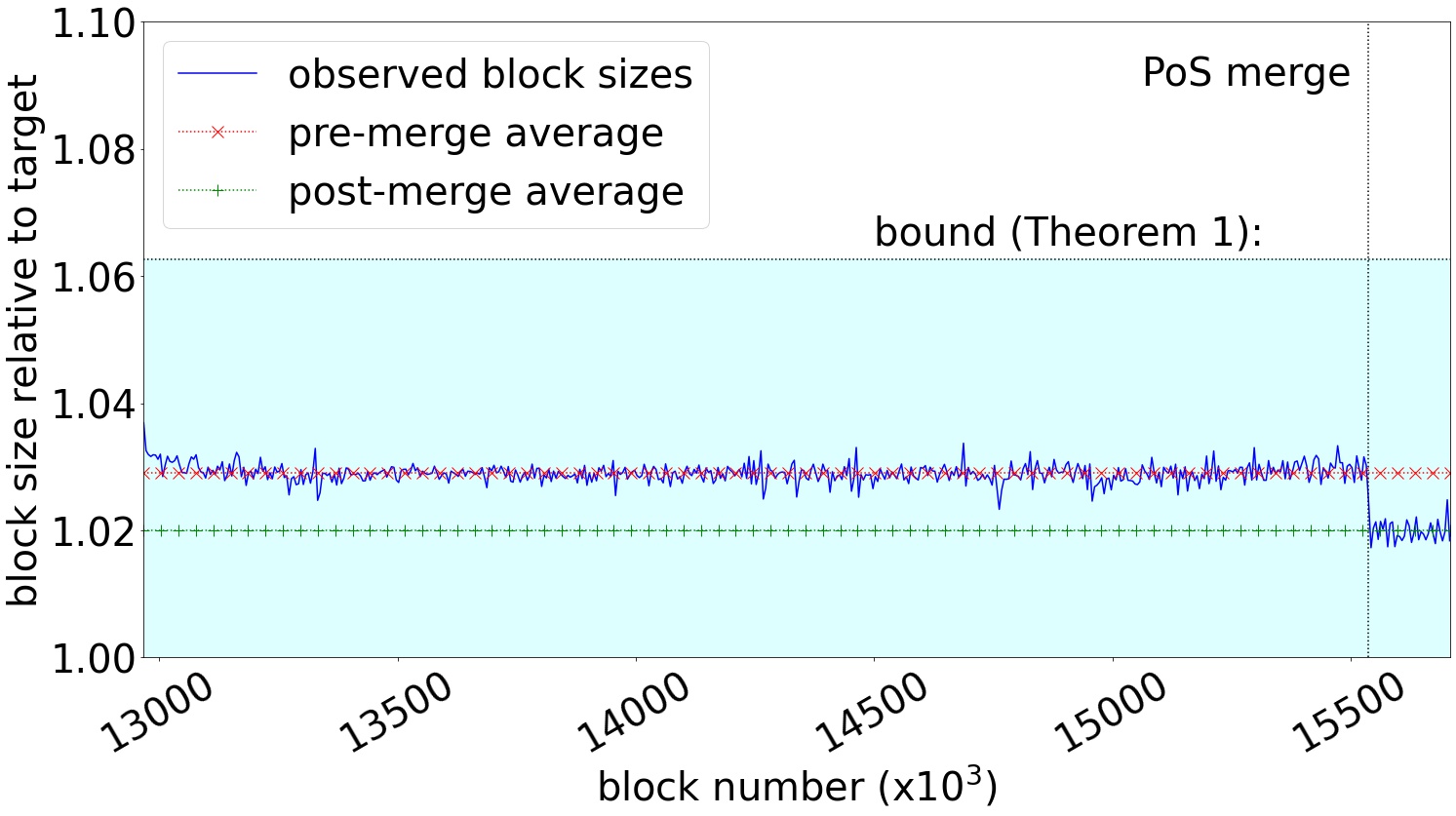}
    \caption{Evolution of the relative block size since EIP-1559: the blue line without marks depicts the observed average block size over batches of 5000 consecutive blocks. The red and green lines with marks depict the averages over the periods before and after Ethereum's switch to proof-of-stake, respectively. The colored region indicates the range of potential long-run averages covered by the bound of \Cref{thm:main}.}
    \label{fig:empirical_averages}
\end{figure}

In the current work, we investigate whether optimality is possible in fee markets that exhibit non-convergent behavior. 
Specifically, we show that the default EIP-1559 mechanism is \textit{approximately optimal} even if the block sizes are chaotic. We find that, unless market conditions are such that the base fee converges to a fixed point, EIP-1559 1) overshoots the target but 2) by at most $\approx$$6.27\%$. 
These results hold \textit{regardless} of the specific market conditions beyond convergence, or the nature of the block creation protocol. This suggests that we can still analyze the system
even if it does not reach an equilibrium, which is a very stringent condition to be met in practice. Furthermore, our results allow us to quantify the maximum degree to which excessively large blocks impact nodes with limited processing power.
We have validated our theoretical results using historical data, as displayed in \Cref{fig:empirical_averages} (see \Cref{sec:empirical_evaluation}): since its introduction, blocks in EIP-1559 were initially $\approx$$2.9\%$ larger than the target, and after a major change to Ethereum's block creation protocol\footnote{In particular, Ethereum's consensus protocol switched form Proof-of-Work to Proof-of-Stake (PoS), which, as a by-effect, caused the time between the creation of new blocks to become constant.} (the PoS ``Merge'' \cite{themerge}) this overshoot dropped to $\approx$$2.0\%$. Both of our main theoretical findings -- blocks overshoot the target, but to a limited degree -- have therefore been borne out in practice. Moreover, the persistence of excessive block sizes throughout the observation period suggests that this is not merely a honeymoon effect \cite{Rei21}. We also observe the tightness of our bound in a wide range of simulation experiments (see \Cref{app:experiments}). \par

In practice, the baseline variant of EIP-1559 is a linear approximation of an exponential update rule that is computationally inefficient to implement \cite{Fei22}. As a further contribution, we investigate the average-case performance of this ideal mechanism, deemed \emph{exponential} EIP-1559. Our analysis suggests that {exponential} EIP-1559 \textit{always} achieves the long-term average target, even if the base fee does not converge to the market-clearing price. We find that the manner in which EIP-1559 approximates an exponential function creates the observed overshoot -- as such, this suggests an interesting direction for future protocol updates.\footnote{In \Cref{app:proper}, we discuss further such designs including a recent proposal that sets base fees using the principles behind Automated Market Makers (AMMs) \cite{vitalik2021amm} and general dynamic posted price mechanisms \cite{ferreira2021dynamic}.}

\paragraph{Outline:} The outline of our work is as follows. After a discussion of the context of our work (\Cref{sec:background}) and our model of a blockchain economy (\Cref{sec:model}), we present a unifying framework for the dynamics of different transaction fee market mechanisms, including the default EIP-1559 mechanism, its exponential variation, and alternative proposals (\Cref{sec:mechanisms}). In \Cref{sec:long_nerm_analysis}, we present our formal analysis. 
In \Cref{sec:discussion}, we discuss the generality of our results, and \Cref{sec:conclusions} concludes the paper.




    
    


\section{Background and Related Work}\label{sec:background}

In this section, we provide a high-level description of Ethereum's fee market and EIP-1559, and present some concepts that we have not discussed previously. We conclude the section with an overview of related work on fee markets.

\subsubsection*{Ethereum: }
Ethereum is a cryptocurrency platform that supports \textit{smart contracts}, i.e., software programs that are executed in a decentralized network. 
Ethereum's global state consists of the state of all smart contracts and the amount of Ethereum's native cryptocurrency token -- \textit{Ether} or ETH -- in each user account.
The purpose of Ethereum transactions is to transfer ETH from one account to another, or to create or call a smart contract. A selection of nodes have the ability to periodically group transactions into a new \textit{block} and broadcast it to the network. Although the exact nature of these nodes depends on the consensus mechanism, we will refer to such nodes as ``\textit{miners}'' for brevity.
Each block points to a previous block, forming a \textit{blockchain}. 
Each operation in a transaction consumes gas, and the amount of ETH that a user is willing to pay for each unit of gas depends on the user's \textit{valuation} of the transaction -- i.e., how much utility she expects to derive from the transaction's inclusion on the blockchain. Before EIP-1559, users would specify a {gas price} for each transaction that determines the amount of ETH spent per unit of gas. Demand for gas fluctuates over time, e.g., due to temporal patterns and events such as NFT drops, so non-expert users were forced to guess appropriate gas prices from recent data.


\subsubsection*{EIP-1559: }
EIP-1559 simplifies Ethereum's original fee market design through the use of a dynamically adjusted {base fee}. The base fee at each time serves as a \textit{posted price} that users need to pay to have their transaction processed at the next block. 
When users pay a transaction fee, an amount of ETH equal to the base fee is burned -- however, a small amount of ETH can be awarded to the miner by the user in the form of a \textit{miner's tip}. Without the miner's tip, miners would have no incentive to process transactions, which could cause them to create empty blocks instead. The base fee is continuously updated to reflect changing market conditions: if blocks are larger (smaller) than a fixed \textit{target} size, then the base fee is increased (decreased) to reduce (increase) demand. The target has been set to roughly equal the maximum block size before EIP-1559 (i.e., 15M gas) -- meanwhile, the maximum block size has been increased to 30M gas (i.e., twice the target). A higher maximum block size (relative to the target) increases the risk that nodes are overwhelmed during demand bursts.

\subsubsection*{Related work: }

In \cite{Rou20,roughgarden2021transaction}, three desirable properties for transaction fee markets are investigated: whether (1) users are incentivized to bid their true valuation, (2) miners are incentivized to follow the protocol's inclusion rule, and (3) miners and users have no incentive to form cartels to subvert the protocol. It is shown that EIP-1559 always satisfies property (1) and (3), and (2) only when demand is ``stable'' \cite{roughgarden2021transaction}. Ethereum's original fee market (being a first-price auction) does not satisfy property (1), whereas property (3) would not hold if the base fee were awarded to miners instead of burned. The compatibility of these properties under general conditions is further explored by \cite{chung2021foundations,shi2022can} and \cite{Gaf22} in the context of transaction fee markets and by \cite{Mil22} in the context of NFT auctions.

In \cite{monnot2020ethereums,Leo21}, the behavior of the base fee under stable market conditions is investigated -- it is found that the base fee exhibits Li-Yorke chaos \cite{Leo21}, which in practice results in the prevalence of sequences of alternating full and empty blocks. This behavior was later confirmed to occur in practice \cite{Rei21}. In \cite{ferreira2021dynamic}, the \textit{social welfare} of fee market mechanism is investigated, and two alternative mechanisms to EIP-1559 are proposed that perform better on this metric. In \cite{liu2022empirical}, the impact of EIP-1559 on various user-centric measures such as average transaction fees, waiting times, and consensus security is investigated. In \cite{diamandis2022dynamic}, an extension of EIP-1559 is considered in a setting in which the base fee depends on the availability of multiple fungible resources (i.e., beyond gas use).
Finally, fee market design has been studied for other cryptocurrency platforms, e.g., Bitcoin \cite{basu2019stablefees,lavi2022redesigning}.

The question of whether the fee market ensures that the long-term average block size indeed equals the target size has not been considered in these works, although \cite{liu2022empirical} finds that the average size of blocks as measured in terms of the network load (which does not capture the gas use of, e.g., smart contract function calls) has increased from 64.05 to 78.01 kB. 

\section{Model and Notation}\label{sec:model}

In this section, we introduce our notation to describe transaction fee markets mathematically. We have EIP-1559 in mind (cf. \Cref{sec:background}), but the description applies to variations of EIP-1559 and other related mechanisms as well.


\paragraph{Base fee and target block size:}
The main element of EIP-1559 like transaction fee markets is a dynamically adjusted \emph{base fee}, $b_n$, that is updated after every block, $B_n, n\ge0$. The goal of the mechanism is to update the base fee in such a way that blocks achieve a pre-defined target block size. Let $T$ denote the target block size and let $kT$ denote the maximum block size, for some integer $k\ge 1$. Currently, in the Ethereum blockchain, $k$ is set at $k=2$, i.e., the target is equal to half the maximum block size. 
\paragraph{Users (transactions) and valuations:}
Users (transactions) arrive to the pool according to a stochastic process. Without any loss of generality, we assume that each transaction uses $1$ unit of gas, and use the random variable $N_n$ to describe the number of transactions that arrive between two consecutive blocks $B_n,B_{n+1}$ for $n\ge0$. 
Let $\lambda_n = \mathbb{E}(N_n) / T$ be the ratio of the expected value of $N_n$ to the target $T$.
We make no assumptions about the distribution of $N_n$ beyond it having a finite mean, i.e., $ \lambda_n < \infty $.
To avoid trivial cases, we will assume that $\lambda_n > 1$, i.e., that the arrival rate is larger than the target block size. For the theoretical analysis, we will assume that users leave the pool if their transaction is not included in the next block and return according to the specified arrival process.\footnote{This assumption only reduces unnecessary complexities in the analysis and is relaxed in the simulations without significant effect in the results.} Whenever necessary, we will index users (transactions) with $i,j\in\mathbb N$.
\paragraph{Users' valuations:}
Each transaction, indexed by $i \in \mathbb N$, has a valuation, $v_i$. Valuations at time $t$ are drawn from a distribution function $F_n$ with support included in $\R^+$. Typically, the support of $F_n$ is bounded, i.e., there exists a maximum possible valuation $M\gg 0$.\footnote{While unbounded valuations are unrealistic for practical purposes, we note that our results hold even for such cases.} We will write $\ol{F}_n(x):=1-F_n(x)$ to denote the so-called \emph{survival function} of the distribution $F_n$. For simplicity, we will assume that $\lambda_n \equiv \lambda$ and $F_n \equiv F$ for all $n \geq 0$, i.e., that $\lambda$ and $F$ are independent of the block height. We discuss how our results are straightforwardly generalized to a setting with time-dependent distributions in \Cref{sec:discussion}.
\paragraph{Bids and tips:}
User bids in EIP-1559 consist of two elements: (1) the \emph{max fee}, $f$, which is the maximum amount per gas unit that the user is willing to pay for their transaction to be included and (2) the \emph{max priority fee}, $p$, which is the maximum tip per gas unit that the user is willing to pay to the miner who includes their transaction. We assume that users bid truthfully and rationally, i.e., a user with valuation $v_i$ will bid $(f,p)=(v_i,\epsilon)$ where $\epsilon>0$ is the minimum amount that covers the miners' cost to process the transaction. Combining with the above, this generates the inclusion requirement: $\text{miner's tip}=\min{\{f-b_n,p\}}\ge\epsilon.$

\paragraph{Block Sizes:}
Let $g_n := g(b_n)$ denote the number of transactions that get included in block $B_n$ when the base fee is equal to $b_n, n\ge0$. Given a number of transactions $N_n = n$ and valuations $v_{1},\ldots,v_{n}$, we have that
$g(b_n) = \min\left\{kT, \sum_{i=1}^{n} {\bf 1}(v_i \geq b_n)\right\},$
i.e., the size of the block is equal to minimum between the block limit, $kT$, and the number of transactions whose max fee exceeds the base fee.\footnote{To simplify notation, we henceforth assume that each user's priority fee, $p$, is equal to the miners' breaken even cost $\epsilon$. Equivalently, we only consider transactions that miners are willing to include and hence, we apply the indicator to $b_n$ instead of $b_n+\epsilon$.} We note that ${\bf 1}(v_i \geq b_n)$ has a Bernoulli distribution with probability $\mathbb{P}(v_i \geq b_n) = \ol{F}(b_n)$ of observing $1$.
For our analysis, we will consider the mean-field approximation of the stationary demand, which results in 
\begin{equation}
\begin{split}
g(b_n) & = \min\left\{kT, \sum_{i=1}^{N_n}\nolimits {\bf 1}(v_i \geq b_n)\right\} = \min\left\{kT, \lambda T \ol{F}(b_n)\right\}
\end{split}
\label{eq:gas_use}
\end{equation}
We denote the market-clearing price, i.e., the base fee for which $g(b^*) = T$, by $b^*$. From \eqref{eq:gas_use}, we observe that $b^* = \bar{F}^{-1}\l1/\lambda\r$.
Equation \eqref{eq:gas_use} also implies that $\lim_{b_n \downarrow 0} g(b_n) =kT$ since $\lambda>k$ by assumption, and $\lim_{b_n \rightarrow \infty} g(b_n) = 0$. To reflect practical settings, we will assume that $b_n$ cannot become negative.


\section{Fee Market Mechanisms: Base Fee Update Rules}\label{sec:mechanisms}
Base fee update rules (BFURs), are functions $h : (0,\infty) \rightarrow (0,\infty)$.
Their goal is to efficiently regulate block sizes via updates in the base fee, $b_{n+1} := h(b_n)$ with $b_0>0$. Intuitively, the base fee increases (decreases) whenever blocks are more (less) than the target and remains constant otherwise. \par

\paragraph{Design goal:} Achieving the target block size in \emph{each} block, however, turns out to be a very difficult \cite{liu2022empirical,Rei21} or even theoretically impossible goal \cite{Leo21}. To obtain a more tractable objective, it is still reasonable to ask whether the target block size is achieved \emph{on average}. In symbols, let $G_N := \frac{1}{N} \sum_{n=1}^N g_n$ denote the average block size until block $N>0$. Then, this requirement suggests that \[\lim_{N \rightarrow \infty} G_N = T.\] 

\subsection{Proper Base Fee Update Rules}\label{sub:proper}
To avoid pathological cases, a base fee update rules $h$ needs to satisfy some minimal regularity conditions. These our outlined in \Cref{def:proper}.\footnote{Apart from the current base fee, $b_{n}$, a base fee update rule may also depend on other parameters, $\theta_n$, such as the target block size (time-independent) or the block size and efficient gas prices at time $n$ (time-dependent). Whenever irrelevant, we will omit such parameters from the description of $h$. In \Cref{app:proper}, we provide the generalized counterpart of \Cref{def:proper} that accounts for such dependencies.}

\begin{definition}[Proper Base Fee Update Rules (\props)]\label{def:proper}
Let $T$ be the target block-size. An update rule $h : (0,\infty) \rightarrow (0,\infty)$ is called \textit{proper} if it satisfies the following
\begin{itemize}[leftmargin=*,align=left]
\item[(A.1)] \emph{non-divergence}: $h(b_n) \geq b_{n}$ if $g(b_n) \geq T$  and $h(b_n) \leq b_{n}$ if $g(b_n) \leq T$.
\item[(A.2)] \emph{bounded relative differences}: there exists an $\alpha\geq1$ such that $\alpha^{-1} b_n \leq h(b_n) \leq \alpha b_n$.
\end{itemize}
\end{definition}
Assumption (A.1) ensures that base fee updates are in the right direction, i.e., that the base fee does not decrease (increase) whenever the current block size is more (less) than the target. 
Assumption (A.2) excludes update rules with potentially unbounded updates. To include more general update rules, (A.2) can be relaxed to (the equivalent in flavor): 
\begin{itemize}[leftmargin=*,align=left]
    \item[\textit{(A.2$'$)}] \emph{there exists $\alpha>1$ and $\beta>0$ such that $\alpha^{-1}b_n - \beta \leq h(b_n) \leq \alpha b_n + \beta$.}
\end{itemize}
\props{} have the desirable property that they generate a bounded sequence of base fees. This is established in \Cref{lem:proper} which can be proved by induction (cf. \Cref{sub:omitted}).

\begin{lemma}\label{lem:proper}
If $h$ is a \prop, and $0<b_0<\infty$, then
$$\min\{b_0,\alpha^{-1} b^*\} \leq b_n \leq \max\{b_0, \alpha b^*\} \;\;\textnormal{ for all }\;\; n\geq 0.$$
\end{lemma}

\subsection{Examples of PBFURs}

\subsubsection{EIP-1559:}
In the EIP-1559 transaction fee market \cite{monnot2020ethereums,Leo21,Rei21}, the base fee, $b_n$, is updated after every block (where blocks are indexed by their block height, $t>0$) according to the following equation 

\begin{equation}\label{eq:eip1559}
b_{n+1}=b_n \l 1+d\cdot\frac{g_n-T}{T}\r, \qquad \text{for any } t\in\mathbb N,\tag{EIP-1559}
\end{equation}
where $d$ denotes the \emph{adjustment quotient} (or step-size or learning rate), currently set by default at $d=0.125$. It will be convenient to use the notation \mbox{$y_n:=\frac{g_n-T}{T} \in [-1,1]$}, for the \emph{normalized deviation at block $t$}, and 
\begin{equation}\label{eq:average}
    G_N:=\frac{1}{N}\sum_{n=1}^N\nolimits g_n
\end{equation} for the \emph{average block size up to block $N$}. It is immediate to check that \eqref{eq:eip1559} satisfies both A.1 and A.2 and is, thus, a \prop.

\subsubsection{Exponential EIP-1559:}
Instead of the \eqref{eq:eip1559} updates, we may consider the exponentially weighted updates
\begin{equation}\label{eq:exponential}
b_{n+1}=b_n \l 1+d\r ^{\l\frac{g_n-T}{T}\r}, \qquad \text{for any } t\in\mathbb N.\tag{EXP-1559}
\end{equation}
The standard \eqref{eq:eip1559} update rule is the linear approximation (in the Taylor expansion of the function $d\mapsto (1+d)^{y_n}$) of the update rule in equation \eqref{eq:exponential}. Again, it is immediate to check that \eqref{eq:exponential} satisfies non-divergence (A.1) and bounded relative differences (A.2) and is, thus, a \prop. Using the \href{https://en.wikipedia.org/wiki/Bernoulli\%27s_inequality#Generalization_of_exponent}{generalized Bernoulli inequality}, it is also straightforward to show that the updates of \Cref{eq:exponential} are always less aggressive than (but in the same direction as) the updates of \Cref{eq:eip1559}. In other words, if a block is congested ($y_n>0$), then the next base fee will increase with both methods, but it will be higher with \eqref{eq:eip1559}. Similarly, if a block is not congested ($y_n<0$), then the next base fee will decrease with both methods, and it will be lower with \eqref{eq:exponential}.\footnote{Ethereum researchers \cite{Fei22}, also study an alternative exponential EIP-1559 form which relies on the exponential approximation $1+dy_n\approx e^{dy_n}$ for $dy_n$ small enough. Again, it is immediate to see that this rule satisfies (A.1) and (A.2) and is, thus, a \prop.}

\section{Analysis: Bounds on Average Block Sizes}\label{sec:long_nerm_analysis}

\subsection{EIP-1559}\label{sub:bounds}
In this section, we are interested to obtain lower and upper bounds on the long-term average block sizes, $G_N, N\ge0$, generated by the EIP-1559 update rule, cf. equation \eqref{eq:eip1559}. Our main result is summarized in \Cref{thm:main}.

\begin{theorem}\label{thm:main}
Let $(g_n)_{t>0}$ denote the sequence of block sizes generated by the base fee $(b_n)_{t>0}$ of the EIP-1559 update rule \eqref{eq:eip1559} with learning rate $d\in(0,1)$ for an arbitrary valuation distribution on $\mathbb R_+$. Then, the long-term average block size $\lim_{N\to+\infty}G_N$ satisfies
\begin{equation}\label{eq:main}
T\le \lim_{N\to+\infty}G_N\le \lt 1-\frac{\ln{(1+d)}}{\ln{(1-d)}} \rt^{-1} 2T.
\end{equation}
\end{theorem}

\begin{wrapfigure}[17]{r}{0.5\linewidth}
 \vspace*{-0.7cm}
\begin{tikzpicture}[>=stealth,scale=0.68]
    \begin{axis}[
        xmin=0,xmax=0.5,
        ymin=0.5,ymax=0.7,
        axis line style=->,
        xlabel={learning rate $d$},
        ylabel={upper bound (factor of T)},
        ]
        \addplot[blue]
        expression[domain=0:0.1,samples=50]{0.253*x+0.5};
        \addplot[blue]
        expression[domain=0.1:1,samples=100]{(1-ln(1+x)/ln(1-x))^(-1)}
        node at (0.2,0.65) {$y(d)=\lt1-\frac{\ln{(1+d)}}{\ln{(1-d)}}\rt^{-1}$}; 
    \end{axis}
\end{tikzpicture}
    \caption{Upper bound (scaling factor of maximum block size if $k=2$) in equation \eqref{eq:main} of \Cref{thm:main}. The upper bound grows almost linearly for the relevant values of $d$.}
    \label{fig:upper_bound}
\end{wrapfigure}
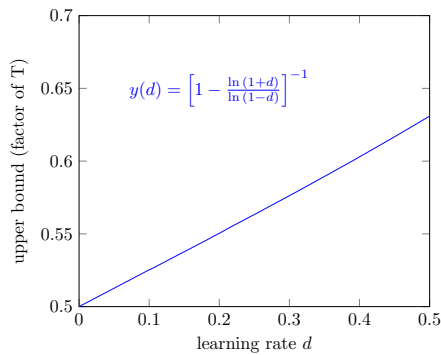

In words, \Cref{thm:main} implies that the EIP-1559 update rule either meets or \emph{slightly} overshoots the target of $T$. The extent of possible overshooting (see upper bound in \eqref{eq:main} of \Cref{thm:main}) depends on the choice of the learning rate $d$. For instance, at the current default level of $d=0.125$, this yields a bound of approximately $1.0627$ or $106\%$ of the block size $T$. The upper bound for values of $d\in(0,0.5]$ (which are of practical interest) is visualized in \Cref{fig:upper_bound}. As we see, this bound grows approximately linearly in $d$. This pattern continues for larger values of $d$ till eventually growing exponentially fast and approaching $2$ (or $200\%$ of $T$) in the limit $d \uparrow 1$ (but such values are, at least currently, only of theoretical interest). The proof of \Cref{thm:main} utilizes the following upper and lower linear bounds on the natural logarithm.

\begin{lemma}\label{lem:linear}
(i) Let $x>-1$. Then it holds that \[\ln{(1+x)}\le x,\]
with equality if and only if $x=0$. (ii) Let $d\in(0,1)$ and let $|x|\le d$. Then, it holds that 
\[\ln{(1+x)}\ge \alpha x+\beta,\]
with $\alpha =\frac{1}{2d}\cdot\lt\ln{(1+d)}-\ln{(1-d)} \rt$ and $\beta=\frac{1}{2}\cdot\lt\ln{(1+d)}+\ln{(1-d)}\rt$.
\end{lemma}

Using \Cref{lem:linear}, we can now prove \Cref{thm:main}.
\begin{proof}[Proof of \Cref{thm:main}]
Let $d\in(0,1)$. By taking the natural logarithm on both sides of equation \eqref{eq:eip1559}, we obtain that
\begin{equation*}
\ln(b_{n+1}/b_n) = \ln{(1 + d\frac{g_n-T}{T})}
\end{equation*}
Recall that $\frac{g_n-T}{T}\in [-1,1]$ by definition. Thus, by applying \Cref{lem:linear} on $\ln{(1 + d\frac{g_n-T}{T})}$ with $d\frac{g_n-T}{T}\in[-d,d]$, we obtain that 
\[\alpha d \cdot \frac{g_n-T}{T}+\beta\le \ln(b_{n+1}/b_n)\le d\cdot \frac{g_n-T}{T},\]
with $\alpha, \beta$ as above.
Combining the above and solving for $g_n$, we obtain that 
\[\frac{T\cdot\ln{(b_{n+1}/b_n)}}{d}+T\le g_n \le \frac{T\cdot\ln{(b_{n+1}/b_n)}}{\alpha d}+\l 1-\frac{\beta}{\alpha d} \r T.\]
Observe that if we sum up all terms from $1$ to $N$, then the term involving $\ln{(b_{n+1}/b_n)}$ on both sides telescopes to
\[\sum_{n=1}^N\ln{(b_{n+1}/b_n)}=\sum_{n=1}^N\l\ln{b_{n+1}}-\ln{b_n}\r=\ln{b_N}-\ln{b_1}.\]
Thus, summing up all terms from $1$ to $N$ in the previous inequality, and using the notation $G_N=\frac{1}{N}\sum_{n=1}^Ng_n$, we obtain that
\begin{equation}\label{eq:aux_up}
G_N\le \frac{T\l\ln{b_N}-\ln{b_1}\r}{N\alpha d}+\l 1-\frac{\beta}{\alpha d}\r T\,,
\end{equation}
and 
\begin{equation}\label{eq:aux_low}
\frac{T\l\ln{b_N}-\ln{b_1}\r}{Nd}+T\le G_N\,,
\end{equation}
for the upper and lower bounds respectively. Concerning the term on the right hand side of the upper bound, observe that after some standard algebraic manipulation, we can write $ \l 1- \frac{\beta}{\alpha d}\r \cdot T=\lt 1- \frac{\ln{(1+d)}}{\ln{(1-d)}}\rt^{-1} \cdot 2T$. To conclude observe that $\lim_{N\to+\infty}b_N<M$ for some $M>0$, since $(b_n)_{n\ge0}$ is bounded by \Cref{lem:proper}.\footnote{The limit is also bounded from above if transactions are no longer included by miners when the base fee becomes so high that the computational cost of processing transactions outweighs any potential miner's tip.} This implies that 
\[
\lim_{N\to+\infty}\frac{T\ln{b_N}-\ln{b_1}}{N}=0.
\]
Thus, taking the limit $N\to+\infty$ on both sides of the inequalities in \eqref{eq:aux_up}, \eqref{eq:aux_low}, we obtain that 
\[
T \le \lim_{N\to+\infty} G_N \le \lt 1- \frac{\ln{(1+d)}}{\ln{(1-d)}}\rt^{-1} \cdot 2T,
\]
as claimed. \qed
\end{proof}


\subsection{Visualizations: Bifurcation Diagrams and Long-Term Averages}

To gain more intuition on the previous results, we proceed to visualize the individual trajectories of the base fee dynamics and the resulting block sizes for a simulated demand realization. This is done in the \emph{bifurcation diagrams} of \Cref{fig:bifurcation_diagrams}. 
\begin{figure}[t]
    \centering
    \includegraphics[width=0.96\linewidth]{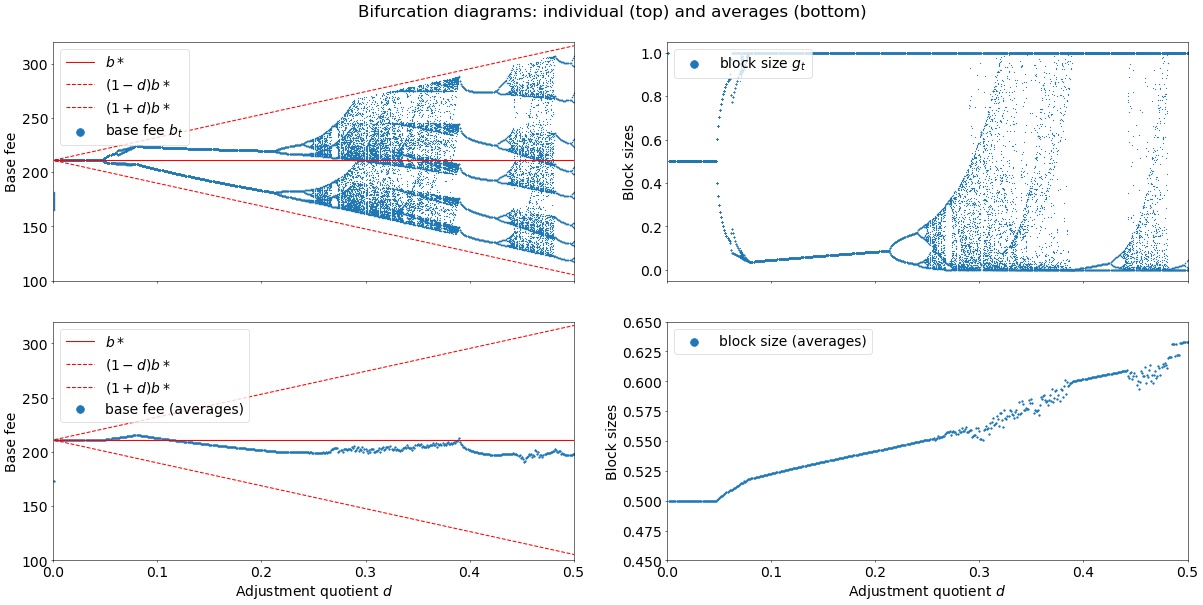}
    \caption{EIP-1559: Individual trajectories ($100$ iterations after skipping $200$ iterations) of the base fee dynamics (top left) and block sizes (top right) for various values of the adjustment quotient, $d$, in $(0,0.5]$ (horizontal axis). The bottom panels show averages of the trajectories in the top panels. Demand (users-valuations) has been simulated from an exponential distribution on $[205,+\infty)$ with mean $\mu=210$ and standard deviation $\sigma=5$. The results are robust to different distributions and initializations (currently $b_0=100$), cf. \Cref{app:experiments}. Despite the chaotic behavior of the individual trajectories (top panels), the long-term averages (bottom panels) exhibit predictable and mathematically tractable patterns (cf. upper bound of block-sizes in \Cref{fig:upper_bound_main}).  Note: Unlike \Cref{fig:empirical_averages}, the scale of the $y-$axis in the block-size panels (top-right and bottom-right) is between $0$ (empty block) and $1$ (full block) and the target is equal to $0.5$.}
    \label{fig:bifurcation_diagrams}
    \vspace{-0.4cm}
\end{figure}

\begin{figure}[t]
    \centering
    \includegraphics[width=0.8\linewidth]{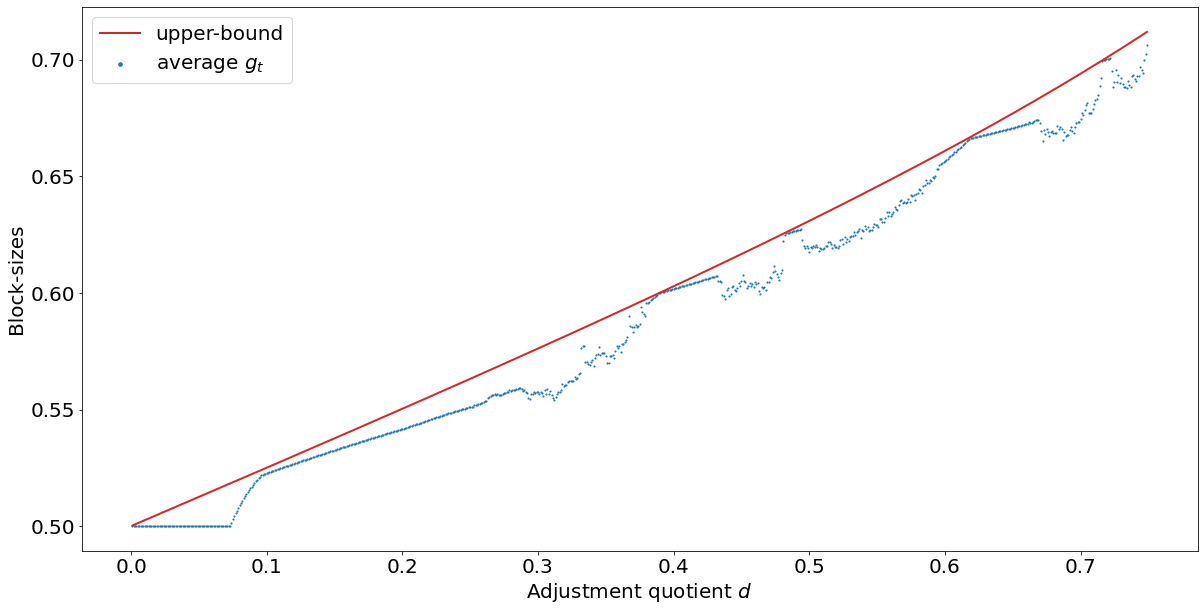}
    \caption{Estimated upper bound (red line) of the average block-sizes as given by equation \eqref{eq:main} in \Cref{thm:main}. The blue dots represent the average block-sizes for normally distributed simulated demand (user-valuations) with mean $\mu=210$ and standard deviation $\sigma=2.5$ at different values of the adjustment quotient, $d$, (horizontal axis). The long-term average block sizes grow linearly with $d$ consistent with the estimated upper bound. Moreover, the upper bound is \emph{tight}. Qualitatively equivalent results obtain for arbitrary demand distributions, e.g., uniform and gamma with arbitrary parameters (not presented here).}
    \label{fig:upper_bound_main}
\end{figure}

\paragraph{Individual trajectories (top panels)}
The bifurcation diagram in the top left panel shows the individual trajectories of the base fee dynamics (blue dots) for different value of the adjustment quotient, $d$, (horizontal axis). Recall that the default value of $d$ is $0.125$.  To generate the plots, we have drawn the user-valuations from a gamma distribution with mean $\mu=220$ and standard deviation $\sigma=10$. The depicted blue dots show 1000 updates after a burn-in of 500 updates (higher numbers of iterations suggest that the dynamics have converged to the depicted attractors). For low values of $d$ (below $0.08$), we see that the base fee dynamics converge to a single value, close to or exactly at the theoretical optimum, $b^*$. For larger values of $d$ (approximately between $0.1$ and $0.3$), the dynamics oscillate between two values and for most remaining values of $d$ (larger than 0.3 and for a small regime roughly between $0.08$ and $0.1$), the dynamics exhibit chaotic behavior (multiple dots dispersed over the whole interval between the two diagonal red lines). In all cases, the dynamics remain within the bounded region $[(1-d)b^*, (1+d)b^*]$.\par
The bifurcation diagram in the top right panel shows the resulting block sizes. Following similar patterns to the base fee dynamics, block sizes converge to the target value ($0.5$ or half-full) for low values of $d$, oscillate between full and (almost) empty for intermediate values of $d$ and become chaotic for larger values of $d$. \par
In summary, the bifurcation plots in the two top panels show the effect of the adjustment parameter, $d$, on the individual trajectories of the base fees and the block sizes. These plots illustrate how small changes in the adjustment quotient, i.e., in the studied \emph{bifurcation parameter}, can cause dramatic changes in the observable trajectories of both base fees and block sizes.
\paragraph{Averages (bottom panels)}
The bottom panels show the averages of the trajectories that are depicted in the top panels; base fees (bottom left) and block sizes (bottom right). We can see that the base fee slightly undershoots the ideal value of $b^*$ and that the block sizes (slightly) overshoot the target of $0.5$. Moreover, the deviation from this target grows linearly in the adjustment quotient $d$. At the current level, i.e., $d=0.125$, the averages are approximately at $0.53$. 
\paragraph{Tight upper-bound on block sizes:}
The figures in the averages in \Cref{fig:bifurcation_diagrams} are not specific to the simulated demand and generalize to arbitrary demand distributions. \Cref{fig:upper_bound_main} shows the estimated upper-bound (red line) of \Cref{eq:main} and the realized average block-sizes for user-valuations drawn from a normal distribution with mean $\mu=210$ and standard deviation $\sigma=10$. The upper-bound is tight and approximates very well the actual evolution of the block-size averages for various values of the adjustment quotient $d$. Qualitatively equivalent results obtain for arbitrarily parameterized uniform, normal and gamma user-valuation distributions (not presented here).

%
\subsection{Exponential EIP-1559}\label{sub:exponential}
As we show in \Cref{thm:exponential}, \eqref{eq:exponential} achieves time average convergence exactly to the target block sizes of $T$. However, \eqref{eq:eip1559} is more relevant from practical purposes since it requires integer rather than floating point calculations.

\begin{theorem}\label{thm:exponential}
For the dynamical system in equation \eqref{eq:exponential}, it holds that 
\[\lim_{N\to +\infty} G_N=T,\]
i.e., the time average of the block sizes (or block occupancies), $(G_N)_{N\ge1}$, converges to the target value $T$ as the number, N, of updates grows to infinity.
\end{theorem}
The proof of \Cref{thm:exponential} mirrors the steps in the proof of \Cref{thm:main} and is therefore deferred to \Cref{app:appendix}. However, it is worth noting that there is nothing special about $T$; the time-average of the dynamic in equation \eqref{eq:exponential} would converge to any given target block-size (as it appears in the numerator of the exponent) within, of course, the admissible limits. 

\subsubsection{Convergence Rates}
From the proof of \Cref{thm:exponential}, we can also reason about the convergence rate. After $N$ time-steps, we have that the distance $d(G_N,T)$ between $G_N:=\frac{1}{N}\sum_{t=0}^Ng_n$ and the target $T$ is equal to 
\[d(G_N,T)=\frac{T\l\ln{b_N}-\ln{b_1}\r}{2\ln{(1+d)}}\cdot\frac{1}{N}\]
which drops linearly in $N$, i.e., $\mathcal{O}(1/N)$. The constant factor depends on $T, \ln{1+d}$ and the error due to initialization, $\ln{b_N}-\ln{b_1}$. Since $b_N$ is bounded (cf. \Cref{lem:proper}, we also know that $\ln{b_N}$ cannot grow (in absolute value) beyond certain bounds. Thus, all these terms vanish at a rate of $1/N$. Note, that in a similar fashion, we can get similar rates for the linear-EIP1559 update, cf. \Cref{thm:main}. In this case, we have that \begin{align*}
  \frac{T\l\ln{b_N}-\ln{b_1}\r}{2d}\cdot\frac{1}{N} &\le d(G_N,T) \le \frac{T\l\ln{b_N}-\ln{b_1}\r}{2\alpha d}\cdot \frac{1}{N}-\frac{\beta}{\alpha d}\frac{T}{2}\,.
\end{align*}
Again, the convergence rate is linear (by the same reasoning as above for $\ln{b_N}$ and the only thing that changes is the constant error term in the upper bound.

\subsubsection{Visualizations for Exponential EIP-1559}

\Cref{fig:bifurcation_diagrams_exp} shows the same data as \Cref{fig:bifurcation_diagrams} but for the exponential EIP-1559 update rule (cf. \eqref{eq:exponential}). The updates of both base fees (top left panel) and block sizes (top right panel) are more smooth (albeit not entirely non-chaotic). However, as can be seen from the left panels, the base fee dynamics tend to overshoot the actual ideal value. This behavior is slightly dependent on the initialization of the dynamics and is not consistent among all possible simulations in many of which the base fee dynamics meet $b^*$ (not presented here). However, in all cases, the important observation concerns the block sizes and the fact that the target is met exactly (horizontal line at exactly $0.5$ in the bottom right panel) regardless of the initialization and the value of the adjustment quotient $d$. This outcome is consistent with the theoretical prediction of \Cref{thm:exponential}.

\begin{figure*}[!t]
    \centering
    \includegraphics[width=0.96\linewidth]{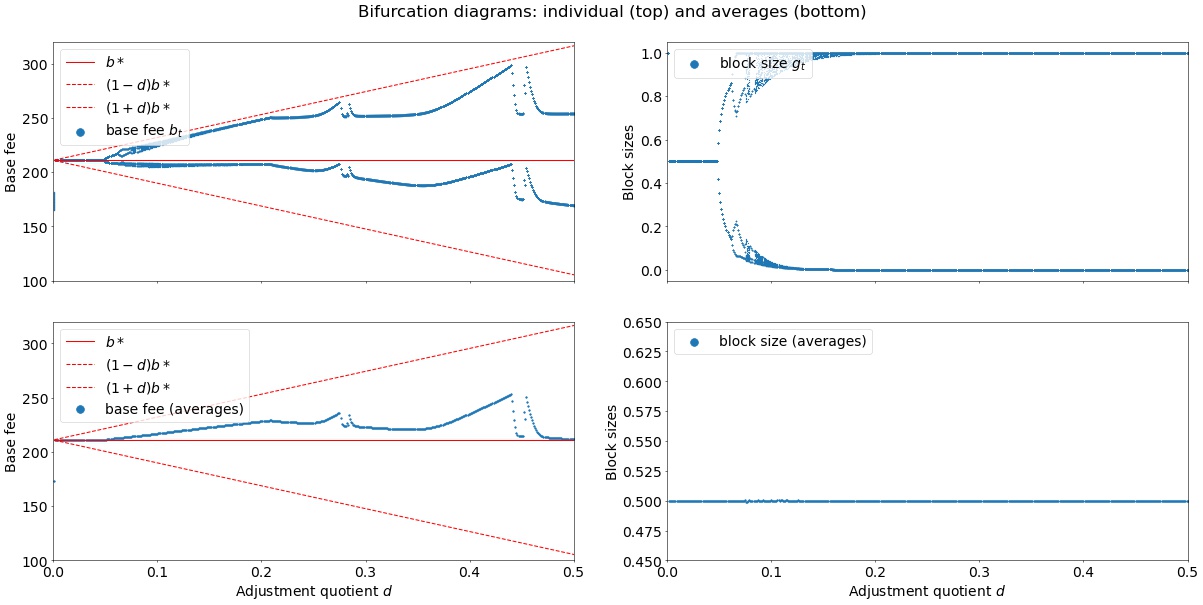}
    \caption{
    Exponential EIP-1559: Individual trajectories of the base fee dynamics (top left) and block sizes (top right) for various values of the adjustment quotient, $d$ (horizontal axis). The panels are the same as in \Cref{fig:bifurcation_diagrams}. Exponential updates result in more regular (individual) trajectories and exactly achieve the block-size target on average. While the exact base fee trajectory is sensitive to $b_0$, it is not clear why the average base fee does not approach the market-clearing price, $b^*$ (horizontal red line).}
    \label{fig:bifurcation_diagrams_exp}
\end{figure*}

\section{Discussion}\label{sec:discussion}

\subsubsection{Overshoot}
The inequality $\ln{\left (1+d\cdot\frac{g_n-T}{T}\right )}\le d\cdot \frac{g_n-T}{T}$ that is used to derive the lower bound of $T$ in \Cref{thm:main} (cf. inequality (i) in \Cref{lem:linear}) holds with equality \textit{if and only if} $d\cdot \frac{g_n-T}{T}=0$, i.e., 
if and only if $g_n=T$. This implies that the long-term average block sizes will be equal to the target, $T$, if and only if the system equilibrates at a fixed point. However, the base fee dynamics are provably chaotic for almost all market conditions (demand and user valuations) that can be met in practice \cite{Leo21}, and, thus, individual blocks will generally deviate from the target. Consequently, the lower bound will hold with \emph{strict} inequality, implying a positive overshoot in EIP-1559 (almost) regardless of market conditions.

\subsubsection{Generality \& Robustness} We note that the proof of \Cref{thm:main} does not in any way rely on $F_n$ or the distribution of $N_n$, but only on the learning rate $d$. As such, our results hold regardless of the exact distribution of valuations or the block creation protocol. In fact, the only technical requirement is that for \Cref{lem:proper}, we need the existence of a market-clearing price $b^*$ to express our bounds. However, if we make the market-clearing price $b^*$ time-dependent i.e., $b^*_n$ instead of $b^*$, we only need to require that there exist $b^*_{\min} > 0$ and $b^*_{\max} < \infty$ such that $b^*_{\min} \leq b^*_{n} \leq b^*_{\max}$ for all $n\geq 0$ to obtain similar bounds. These are loose restrictions in practice. Our results also do not rely on $k$, e.g., if the maximum block size were set to $4$ the target, then this would have no impact on the theoretical bounds. Interestingly, the upper bound in \eqref{eq:main} is less than twice the target \emph{regardless} of how small the target is relative to the maximum. 

Although our theoretical results establish general bounds on the long-term average block size, the precise average values may depend on many factors, including the distribution of user valuations and the block creation protocol. We do observe from \Cref{fig:empirical_averages} that the observed block size averages over 5000-block batches exhibit remarkably consistent behavior both before and after the PoS merge. Since EIP-1559, market conditions have changed considerably: the base fee itself has changed from around 100 GWei (1 GWei = $10^{-9}$ ETH) in Jan.\ 2022 to around 10 Gwei in Aug.\ 2022. However, the average block size per batch has remained around $102.9\%$ of the target throughout this period. Interestingly, the average block size dropped to around $102.0\%$ immediately after the switch to PoS. The reason behind this drop is a stimulating direction for future research. One hypothesis is that before PoS, the inter-block times had an (approximately) exponential distribution, whereas they are constant in Ethereum's PoS protocol \cite{buterin2020combining}. As such, the variance of $N_n$ is smaller and block sizes are more regular. Another hypothesis is that blocks are less congested after the switch due to a decrease in inter-block times from roughly 13 seconds on average \footnote{\url{https://ycharts.com/indicators/ethereum_average_block_time}} to 12 seconds.

\section{Conclusions}\label{sec:conclusions}

In this paper, we have formally analyzed the long-term performance of the standard EIP-1559 transaction fee market mechanism and its closely related variants; exponential EIP-1559 among others. Our findings provide a theoretical justification for the anecdotal evidence that blocks are slightly more full than normal. As our main contribution, we have found that both designs, the baseline EIP-1559 and its exponential variant, are approximately and exactly optimal, respectively, even under the prevailing chaotic conditions in inter-block sizes. Importantly, this implies that these mechanisms can still achieve their goals even if the underlying system does not equilibrate, a condition that is rarely met in practice. The empirical data since the launch of EIP-1559 suggest that our results accurately capture reality: observable average block sizes are within the sharp approximation bounds predicted here and this is, in fact, robust to the underlying protocol functionality including both pre- and post-PoS merge periods. 
\par
Concerning future work, the current paper provides a framework to evaluate the performance and analyze the stability of transaction fee or other related cryptoeconomic mechanisms. Practical use cases suggest that blockchain economies are systems with complex dynamics: when these economies are close to their equilibrium state, they can re-adjust their parameters and self-stabilize. However, once they are pushed further away and/or lose their peg to their fundamentals, they start to spiral away and eventually collapse (e.g., the Terra/Luna crypto network). Determining the limits in which these instabilities emerge already before such mechanisms are launched in practice, is critical to improve their efficiency and avoid future financial catastrophes.

\section*{Acknowledgements}

This research/project is supported in part by the National Research Foundation, Singapore under its AI Singapore Program (AISG Award No: AISG2-RP-2020-016), NRF 2018 Fellowship NRF-NRFF2018-07, NRF2019-NRF-ANR095 ALIAS grant, grant PIE-SGP-AI-2018-01, AME Programmatic Fund (Grant No. A20H6b0151) from the Agency for Science, Technology and Research (A*STAR) and the Ethereum Foundation. It is also supported by the National Research Foundation (NRF), Prime Minister's Office, Singapore, under its National Cybersecurity R\&D Programme and administered by the National Satellite of Excellence in Design Science and Technology for Secure Critical Infrastructure, Award No. NSoE DeST-SCI2019-0009.

\bibliographystyle{plain}
\bibliography{ref}

\appendix
\section{Appendix}\label{app:appendix}

\subsection{Omitted Proofs}\label{sub:omitted}

\begin{proof}[Proof of \Cref{lem:proper}]
We prove the assertion using induction.\\
\textbf{Base case:} clearly, $b_0 \geq \min\{b_0,\alpha^{-1} b^*\}$ and $b_0 \leq \max\{b_0, \alpha b^*\}$ by the definition of $\min$ and $\max$, so the lemma holds for the base case.\\
\textbf{Induction step:} we consider two cases: 
\begin{enumerate}
\item $b_n \leq b^*$, then $b_{n+1} = h(b_n) \geq b_n \geq \min\{b_0, \alpha^{-1}b^*\}$ where the first inequality follows from $h$ being a proper update rule (non-divergence) and the second from the induction hypothesis. Furthermore, \[b_{n+1} = h(b_n) \leq a b_n \leq \alpha b^* \leq \max\{b_0, \alpha b^*\},\] where the first inequality follows from $h$ being a proper update rule (bounded relative differences) and the second from the case assumption.
\item $b_n \geq b^*$: similar as above.\qed
\end{enumerate}
\end{proof}

\begin{proof}[Proof of \Cref{lem:linear}]
Part (i) is well known and follows directly from the Taylor expansion of $\ln{(1+x)}$, i.e., 
\[\ln{(1+x)}=\sum_{n=1}^{\infty}(-1)^{n+1}\frac{x^n}{n}\le (-1)^{1+1}\cdot \frac{x^1}{1}=x.\]
To obtain part (ii), i.e., to lower bound $\ln{(1+x)}$ by a linear function $y(x)=ax+b$ for $|x|\le d$, we only need to require that $y(-d)=\ln{(1-d)}$ and $y(d)=\ln{(1+d)}$. This follows directly from the fact that $\ln{(1+x)}$ is concave, and thus, its graph lies above the line that connects any two of its points for all intermediate values between these points. Solving the system of these two equations, yields
\[y(x)=\frac{\ln{(1+d)}-\ln{(1-d)}}{2d}x+\frac{\ln{(1+d)}+\ln{(1-d)}}{2}\,\;,\]
as claimed. \qed
\end{proof}

\begin{proof}[Proof of \Cref{thm:exponential}]
We start from equation \eqref{eq:exponential} and take the natural logarithm of both sides
\[\ln{\l b_{n+1}/b_n\r}= \l\frac{g_n-T}{T}\r\cdot\ln{(1+d)},\]
which after some rearranging of the terms yields 
\[g_n=
\frac{T}{2}\cdot\frac{\ln{(b_{n+1}/b_n)}}{\ln{(1+d)}}+\frac{T}{2}.\]
Summing up all terms from $1$ to $N$, we obtain that 
\begin{equation}\label{eq:auxiliary_3}
\frac{1}{N}\sum_{n=1}^Ng_n=\frac{T}{2N\ln{(1+d)}}\cdot\sum_{n=1}^N\ln{(b_{n+1}/b_n)}+\frac{1}{N}\cdot\frac{N\cdot T}{2}\,.
\end{equation}
The sum on the left hand side telescopes to 
\[\sum_{n=1}^N\ln{(b_{n+1}/b_n)}=\sum_{n=1}^N\l\ln{b_{n+1}}-\ln{b_n}\r=\ln{b_N}-\ln{b_1}.\]
Thus, \eqref{eq:auxiliary_3} becomes
\[
\frac{1}{N}\sum_{n=1}^Ng_n=
\frac{T\l\ln{b_N}-\ln{b_1}\r}{2N\ln{(1+d)}}+\frac{T}{2}\,.
\]
To conclude observe that $\lim_{N\to+\infty}b_N<M$ for some $M>0$, since $(b_n)_{n\ge0}$ is bounded. This implies that 
\[
\lim_{N\to+\infty}\frac{T\l\ln{b_N}-\ln{b_1}\r}{2N\ln{(1+d)}}=0.
\]
Thus, taking the limit $N\to+\infty$ on both sides of the previous equation, we obtain that 
\[
\lim_{N\to+\infty}\frac{1}{N}\sum_{n=1}^Ng_n=T,
\]
as claimed. \qed
\end{proof}
From the proof of \Cref{thm:exponential}, we can also reason about the convergence rate. After $N$ time-steps, we have that the distance $d(G_N,T)$ between $G_N:=\frac{1}{N}\sum{t=0}^Ng_n$ and the target $T$ is equal to 
\[d(G_N,T)=\frac{T\l\ln{b_N}-\ln{b_1}\r}{2\ln{(1+d)}}\cdot\frac{1}{N}\]
which drops linearly in $N$, i.e., $\mathcal{O}(1/N)$. The constant factor depends on $T, \ln{1+d}$ and the error due to initialization, $\ln{b_N}-\ln{b_1}$. Also these terms vanish at a rate of $1/N$.

\section{General BFURs}\label{app:proper}
In this section, we discuss some general BFURs which have been proposed by the research community or which have appeared in the relevant literature. To allow dependencies on more parameters (other than the current base fee), we first provide some additional definitions and notation. \par

Let $\Theta\subseteq \mathbb R^d$ for some integer $d>0$ denote a \emph{parameter} space. Generalized base fee update rules (GBFURs), are functions $h : (0,\infty)\times \Theta \rightarrow (0,\infty)$. In other words, GBFURs condition their updates on both the current base fee, $b_n$, and on additional time-dependent or time-independent information, e.g., the current block-size, $g_n$, or the initial base fee, $b_0$ etc. The notion of \emph{proper} GBFURs can be adapted accordingly (cf. \Cref{def:proper}).

\begin{definition}[Generalized Proper Base Fee Update Rules (G\props)]\label{def:gproper}
Let $T$ be the target block-size and let $\Theta$ be a parameter space. An generalized update rule $h : (0,\infty)\times \Theta \rightarrow (0,\infty)$ is called \textit{proper} if it satisfies the following
\begin{itemize}[leftmargin=*,align=left]
\item[(GA.1)] \emph{non-divergence}: $h(b_n,\theta_n) \geq b_{n}$ if $g(b_n) \geq T$  and $h(b_n,\theta_n) \leq b_{n}$ if $g(b_n) \leq T$.
\item[(GA.2)] \emph{bounded relative differences}: there exists an $\alpha\geq1$ such that $\alpha^{-1} b_n \leq h(b_n,\theta_n) \leq \alpha b_n$.
\end{itemize}
\end{definition}

Using the above, we can now proceed to the study of general BFURs.

\subsubsection{AMM-based mechanisms:}\label{sub:amm}
Another alternative to \eqref{eq:eip1559} that has been proposed is the AMM-based mechanism which is defined as follows \cite{vitalik2021amm}. Let $x_n$ denote the stochastic process of the \textit{excess gas} issued. The evolution of $x_n$ is given by
\begin{equation}\label{eq:amm_excess_gas}
x_{n+1} = \max \{0, x_n + g(b_n) - T \}
\end{equation}
with $x_0=0$. Then, the AMM base fee update rule is given by
\begin{equation}\label{eq:amm}
b_{n+1} = q e^{q x_{n+1}}\tag{AMM}
\end{equation}
where $q$ is the \textit{adjustment quotient} which determines how the base fee is chosen based on the current excess gas ($q$ is similar to parameter $d$ in the standard EIP-1559 update). From equations \eqref{eq:amm_excess_gas} and \eqref{eq:amm}, we have that 
\begin{equation}\label{eq:amm_b}
b_{n+1}=b_n\cdot e^{q\max{\{-x_n,g(b_n)-T\}}}.\tag{AMM-closed}
\end{equation}

The \eqref{eq:amm_b} base fee update rule depends on both the current base fee, $b_n$, and the current excess gas, $x_n$. Thus, it is a generalized update rule. Since $x_n\ge0$, we have that $\max{\{-x_n, g(b_n)-T\}}=g(b_n)-T$ whenever $g(b_n)\ge T$. This implies that $h(b_n,x_n)\ge b_n$ if this is the case. Similarly, if $g(b_n)<T$, then $\max{\{-x_n, g(b_n)-T\}}<0$ which implies that $b_{n+1}=h(b_n,x_n)<b_n$. Thus, \eqref{eq:amm_b} satisfies (GA.1). Since $e^y$ is continuous for any $y\in \mathbb R$, it is also straightforward to show that \eqref{eq:amm_b} satisfies (GA.2). Thus, \eqref{eq:amm} is a G\props.

\paragraph{Upper bound on average block sizes:}
We next derive an upper bound for the average block sizes with the \eqref{eq:amm} mechanism. Equation \eqref{eq:amm_excess_gas} implies that $x_{n+1} \geq x_n + g(b_n) - T $ or equivalently that $g(b_n) \leq  x_{n+1} - x_n + T$. Thus, 
\begin{align*}
\frac{1}{N} \sum_{n=1}^N g(b_n) &\leq \frac{1}{N} \sum_{n=1}^N \l x_{n+1} - x_{n} \r + T =\frac{1}{N} \l x_{N+1} - x_{1} \r + T
\end{align*}
If $\lim_{N\to\infty}x_N<X$ for some $X$, i.e., if $x_N$ stays bounded, then $\lim_{N \rightarrow \infty}\frac{1}{N} \l x_{N+1} - x_{1} \r = 0 $ and therefore
\[
\lim_{N \rightarrow \infty} \frac{1}{N} \sum_{n=1}^N g(b_n) \leq T.
\]
Alternatively, from equation \eqref{eq:amm_b}, we have that \[\frac{1}{N}\sum_{n=1}^N\max{\left\{T-x_n,g(b_n)\right\}}=T+\frac{1}{qN}(\ln{b_{N}}-\ln{b_{1}}).\]
Again, if $b_N$ is bounded (which is always the case from below by $L=7$ and also from above, under the assumption that users' valuations are bounded), then 
\[\lim_{N\to\infty}\frac{1}{N}\sum_{n=1}^N\max{\left\{T-x_n,g_n\right\}}=T.\]

\paragraph{Sufficient condition for convergence of AMM to $T$.}
If $x_n\ge T$ for some $n\ge0$, then,   \[\max{\left\{\frac{T}{2}-x_n,g_n\right\}}=g_n.\]
To deal with the case $0\le x_n \le T$, one sufficient condition is to assume that for each such $x_n$, it holds $g_n(b_n)>T-x_n$, where    $b_n=\frac{d}{T}e^{\frac{d}{T}x_n}$. To obtain a simple sufficient condition, we may require that \[\min_{x_n\in [0,T]} {g_n(b_n)}>\max_{x_n\in [0,T]}{\{T-x_n\}}=T.\] Using that $g_n$ is decreasing in $b_n$, and that $b_n$ is strictly increasing in $x_n$, this reduces to requiring that $g_n(b_n(T))>T$, where $b_n(T)=\frac{d}{T}e^d$. \par
In other words, if there are sufficient users (i.e., more than $T$) with valuations above $\frac{d}{T}e^d$, then the $\max$ in the previous summation is always resolved by $g_n$ and we obtain convergence to $T$. In symbols, we can formulate the sufficient condition as 
\[\lambda \overline{F}\left(\frac{d}{T}e^d\right)>1\]
or equivalently (if we solve for $d$),
\[de^d<T\overline{F}^{-1}(\lambda/2)\]
where the inequality was reversed, since the inverse of a decreasing function (here $\overline{F}$) is also decreasing.

\begin{wrapfigure}[17]{R}{0.56\linewidth}
\vspace{-0.92cm}
    \centering
    \includegraphics[width=0.9\linewidth, clip=true, trim = 0cm 0cm 0cm 2cm]{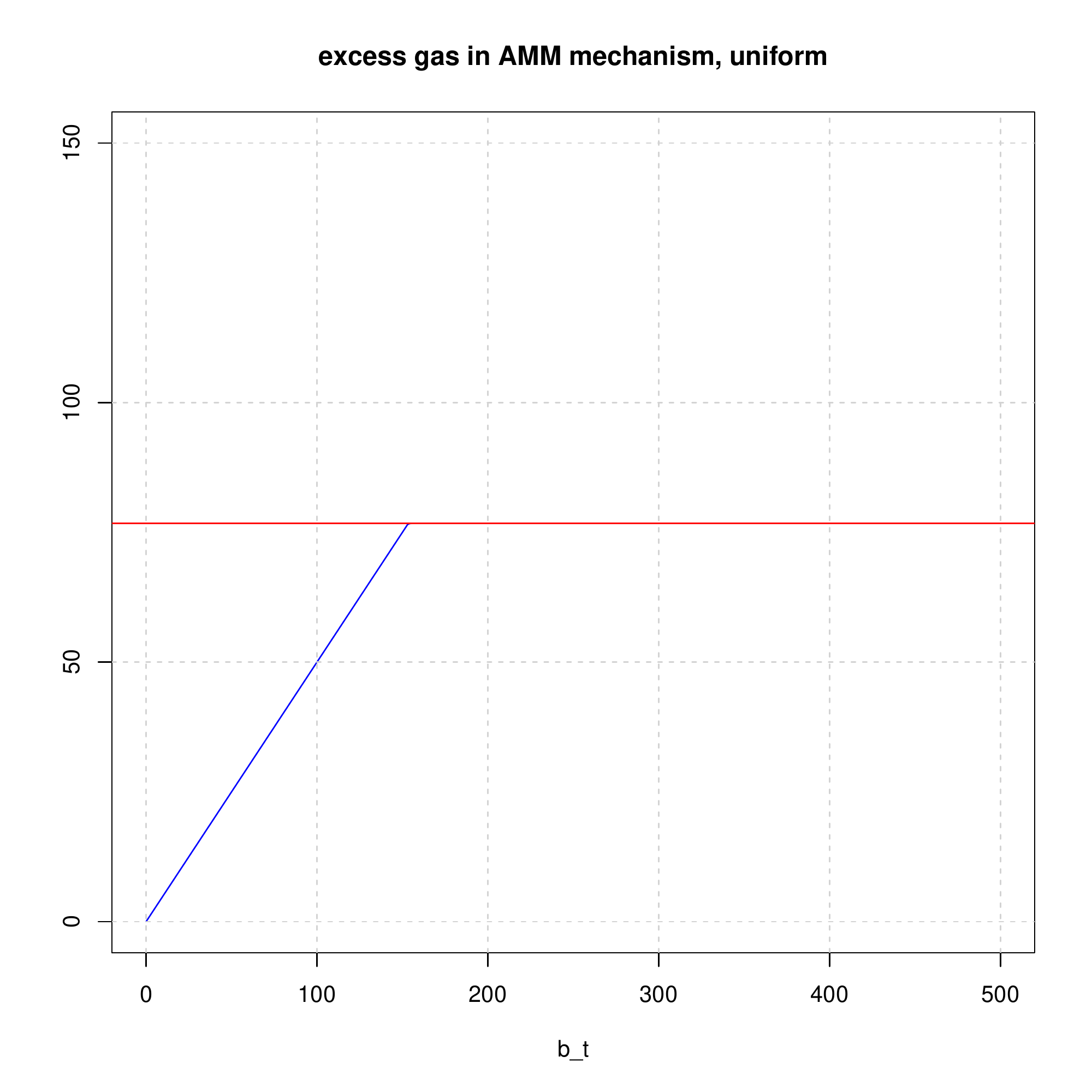}
    \caption{Excess gas over time in the AMM mechanism for uniform demand distribution.}
    \label{fig:amm_excess_gas}
\end{wrapfigure}

Finally, if $b^* = \bar{F}^{-1}\l\frac{T}{2\lambda}\r$, then by \eqref{eq:gas_use} we have that $g(b^*) =T$. Moreover, if for some $n > 0$, it holds that $x_n > 0$ and that $b_n = b^*$, then 
\[
x_{n+1} = \max\{0, x_n\} = x_n,
\]
which implies that there is no change in the excess gas. In this case, it holds that
\[x_n = q^{-1} \ln(q^{-1} b^*).\]
This can be also seen in the example of  \Cref{fig:amm_excess_gas}, where the valuations are drawn uniformly at random from $[200,230]$ with $\lambda=1$, and we use an AMM-based mechanism with $q=\frac{1}{10}$. The market clearing price $b^* = 215$, and the excess gas converges to $10 \cdot \ln (2150) \approx 76.73$.

\subsubsection{Dynamic Posted-Price Mechanisms:}
Another approach to transaction fee mechanisms that has been proposed in the literature \cite{ferreira2021dynamic} is the following valuation-based or \emph{welfare-based} approach
\begin{equation}\label{eq:welfare}
b_{n+1} = \frac{\alpha}{kT} \sum_{i=1}^{g_n} v_{i,n} + (1-\alpha) b_{n}\tag{WEL}
\end{equation}
and the corresponding \emph{truncated welfare-based} alternative:
\begin{equation}\label{eq:truncated}
b_{n+1} = \left\{ \begin{array}{ll} 
\alpha (1+\delta) b_{n} + (1-\alpha) b_{n}, \qquad\text{if } g_n = kT\\  
\dfrac{\alpha}{kT} \displaystyle\sum_{i=1}^{g_n} \min\{v_{i,n},(1+\delta) b_{n}\} + (1-\alpha) b_{n}, & \text{else.} \\  
\end{array} \right. \tag{T-WEL}
\end{equation}
where $kT$ denotes the block capacity. These Dynamic Posted Price mechanisms condition their updates also on the $v_{i,n}$, i.e., the valuations of the included transactions at time $n$. However, their scope is to maximize welfare rather than meeting a specific target $T$. Accordingly, condition $(GA.1)$ is irrelevant in this case. However, assuming that valuations are bounded, it is again straightforward to show that these mechanisms satisfy $(GA.2)$.

\subsubsection{Effective Gas Price Correction Update Rule:}
In some cases, when demand shifts abruptly upwards, the EIP-1559 transaction fee market reverts back to a first price auction (FPA). During such times, e.g., during demand bursts due to NFT drops or other reasons, the prevailing base fee is too low to control the changing demand. This results in a series of full blocks in which all included transactions pay unreasonably high tips (FPA conditions) before the base fee reaches the appropriate level to achieve half-full blocks.\par
To counter the above phenomenon several alternative mechanisms have been proposed \cite{ferreira2021dynamic,Rei21}. However, these mechanisms treat deviations from the target block size equally (treat all kinds of deviations equally) or rely their correction terms again on the previous block-sizes.\par
However, FPA conditions are unlike any other deviations from target block-sizes due to the distinctive characteristic of very high tips to miners, i.e., of very high effective gas prices in comparison to the prevailing base fee. Next, we use this intuition to design a modified EIP-1559 mechanism that precisely aims to mitigate such FPA conditions.\par

Let $m_n$ denote the minimum effective gas price paid at block height $t$, and let  $\delta_n:= m_n - b_n$. By definition, $\delta_n>0$. If $\delta_n$ is large, e.g., $\delta_n>b_n$, then this means that all transactions in block $t$ were selected under first price auction (FPA) conditions. By contrast, if $\delta_n$ is small, e.g., $\delta_n< 0.1b_n$, then this means that the current base fee is correct for at least some users. In the first case, (large $\delta_n$), the mechanism needs to aggressively update $b_n$, whereas in the second case, the mechanism is ok with a normal update. This motivates the following modification of the standard EIP-1559 update rule

\[b_{n+1} =b_n \l 1+d\cdot \frac{g_n-T}{T} \r +\lambda \cdot \relu(\delta_n -\gamma b_n),\]
where $\relu(x):=\max{(0,x)}$. Since $\delta_n=m_n-b_n$, we can rewrite the above as 

\[b_{n+1} =\underbrace{b_n \l 1+d\cdot \frac{g_n-T}{T} \r}_{\text{EIP-1559}} +\underbrace{\lambda \cdot \relu(m_n -(1+\gamma)b_n) \vphantom{\left(\frac{a^{0.3}}{b}\right)} }_{\text{EGP Correction}},\]
where $m_n$ is the minimum effective gas price of block $t$. We will refer to this update rule as the \emph{Effective Gas Price Correction Update Rule} (EGP-CURe). In  comparison to EIP-1559, EGP-CURe has two additional (learnable) parameters: (i) $\gamma$ which is the threshold that triggers the effective gas price correction mechanism and $\lambda$ which calibrates the intensity of the updates in case that the correction mechanism has been triggered. Note that, similar to $g_n$, $m_n$ can be directly estimated from the observable data that is included in block $t$.

\paragraph{EGP-CURe is not PFUR}
If we assume that user valuations are bounded by some constant $M$, then the EGP-CURe update rule, $b_{n+1}:=h(b_n)$ defined above, can be easily shown to satisfy the generalized (A.2$'$) condition. 

\begin{lemma}
Assume that user valuations are bounded by some constant $M$. Then, the EGP-CURe defined above by $b_{n+1}:=h(b_n)$ satisfies generalized bounded relative differences condition, (A.2$'$).
\end{lemma}
\begin{proof}
If $m_n\le(1+\gamma)b_n$, then EGP-CURe degenerates to EIP-1559 which has been shown to satisfy the more restrictive (A.2). Thus, it remains to treat the case $m_n>(1+\gamma)b_n$. In this case, EGP-CURe becomes
\begin{align*}
b_{n+1} & =b_n \l 1+d\cdot \frac{g_n-T}{T} \r +\lambda \cdot (m_n -(1+\gamma) b_n).
\end{align*}

and we have that 
\begin{align*}
b_{n+1}& = b_n \l 1+d\frac{g_n-T}{T} -\lambda(1+\gamma)\r +\lambda m_n \le b_n(1+d-\lambda(1+\gamma)) +\lambda M
\end{align*}
where the last inequality holds because of the bounded valuations assumption, i.e., $m_n<M$ for any $t\ge0$. This means that EGP-CURe satisfies condition (A.2$'$) with constants $\alpha\equiv (1+d-\lambda(1+\gamma))$ and $\beta\equiv \lambda \cdot M$. Showing that $h(b_n)$ is bounded from below is trivial (since the correction term is always nonnegative) which concludes the proof.
\end{proof}
Concerning condition (A.1), it is immediate to show that EGP-CURe satisfies $h(b_n)\ge b_n$ if $g_n\ge T$. However, it does not necessarily hold that $h(b_n)\le b_n$ if $g_n\le T$. For instance, if there are only few transactions in a block (less than half-full block), but they all pay very high tips (above the threshold), then the correction mechanism will be triggered and, under certain mathematical conditions, it will not allow the base fee to decrease (despite a less than full block).\par
Such a case may only be of theoretical interest, since it implies a clearly suboptimal behavior in which users pay unreasonably very high fees to get their transactions included in an otherwise non-full block. Nevertheless, if one wants to exclude such a case, one may additionally impose that the EGP correction term is triggered only if $g_n>T$, i.e., only if the block is more than half-full.

\section{Empirical Evaluation}\label{sec:empirical_evaluation}

In this section we describe the blockchain data from after EIP-1559 used to construct \Cref{fig:empirical_averages}.

\paragraph{Description of Data. }

To collect blockchain data, we have used Google 's BigQuery data service \cite{bigquery}. Our sample includes all blocks $b_n$ whose index $t$ is in $\{t^{-}, \ldots, t^{+}-1\}$ where $t^{-} = 12965000$ and $t^+ = 15200000$. In particular, $t^{+}$ is the index of the first block after EIP-1559, and block $t^+$ was mined on 23 July, 2022. We group the blocks of our sample into $M=447$ consecutive batches, such that each batch consists of $N=(t^{+} - t^{-})/M = 5000$ blocks. For each batch $i=1,\ldots,M$, let $t_i^{-} = t^{-} + N i$ and $t_i^{+} = t^{-} + N (i+1) - 1$ be the indices of the first and last block in batch $i$, respectively. We then define
$
y_i \triangleq \frac{1}{N} \sum_{t = t_i^{-}}^{t_i^{+}} g_n / T
$
as the average relative gas use in the blocks of batch $i$. We note that $g_n / T \in [0,1]$ for all $t$, so it must hold that $y_i \in [0, 1]$ as well. To calculate $y_i$ for each batch, we use the query from \Cref{fig:google_sql_query}. Our approach is fully reproducible -- anyone with a BigQuery account can use the query from \Cref{fig:google_sql_query} to reproduce the data in this section.

\paragraph{Analysis.}
We have visualized the data created by the query of \Cref{fig:google_sql_query} in \Cref{fig:empirical_averages}. The solid blue line in \Cref{fig:empirical_averages} depicts the evolution of $y_i$ across the observation period, whereas the red dashed line depicts the sample average $\mu_y = \frac{1}{M} \sum_{i=1}^M y_i$, which equals $0.5145$. We observe that the batch averages occur in a narrow band around $\mu_y$. From \Cref{fig:upper_bound_main}. we observe that the theoretical upper bound of the block size for $d=0.125$ equals $\l 1 - \frac{\ln(1.125)}{\ln(0.875)}\r^{-1} \approx 0.5313$. This shows that, as predicted by our analysis, the average block sizes both exceed the target relative block size of $0.5$, but not to a degree at which they would exceed our theoretical upper bound.

\begin{figure}
\centering
\framebox[0.96\textwidth]{
\hspace{0.05cm}
\begin{minipage}[t]{0.92\textwidth}
\begin{flushleft}
\footnotesize{\texttt{\querytxta{SELECT}  \\
\tab \querytxta{AVG}(gas\_used/gas\_limit) \querytxta{AS} block\_size, \\
\tab \querytxta{MIN}(number) \querytxta{AS} min\_number, \\
\querytxta{FROM} ( \\
\tab \querytxta{SELECT} \\
\tab \tab gas\_used, \\
\tab \tab   gas\_limit, \\
\tab \tab   number, \\
\tab \tab   \querytxta{FLOOR}(number/\querytxtb{5000}) as bucket \\
\tab   \querytxta{FROM} \querytxtc{{\bq}bigquery-public-data.crypto\_ethereum.blocks{\bq}} \\
\tab   \querytxta{WHERE} number >= \querytxtb{12965000} AND number < \querytxtb{15200000} \\
) \\ 
\querytxta{GROUP BY} bucket \querytxta{ORDER BY} bucket
}}
\end{flushleft}
\end{minipage}
}
\caption{The SQL query used in Google BigQuery to produce the data for \Cref{fig:empirical_averages}. }
\label{fig:google_sql_query}
\end{figure}


\section{Extensive Simulations}\label{app:experiments}

In this part, we present systematic simulations of the \eqref{eq:eip1559} and \eqref{eq:exponential} BFURs. We consider two bifurcation parameters: (1) the adjustment quotient, $d$, and (2) the range of valuations, which is parameterized by $w$ (see \Cref{tab:distributions}).

\paragraph{Adjustment quotient:} For the adjustment quotient, we allow values in the $(0,0.5]$ range, since lower values are not admissible and larger values are not relevant to practice (they result in extremely aggressive updates). We consider three types of distributions for user valuations which are presented in \Cref{tab:distributions}. To define the distributions, we use parameters, $m$ (equal or related to the mean), $w$ (related to standard deviation) and $a$ (only necessary for gamma distributions).
\begin{table}[!htb]
\setlength{\tabcolsep}{12pt}
    \centering
    \begin{tabular}{@{}ll@{\hskip 3pt}ll@{}}
    \textbf{Distribution} & \textbf{Parameters} & \textbf{$(\mu,\sigma^2)$} & \textbf{Support}\\
    \toprule
    Uniform & $m$ & $w^2/12$ & $[m-w/2,m+w/2]$\\
    Normal & $m$ & $w^2/16$ & $(m-w,m+w)$\\
    Gamma & $m-aw$ & $w^2a$ & $(m-aw,+\infty)$\\
    \bottomrule
    {}&{}
    \end{tabular}
    \caption{Distributions of users valuations that have been used in the simulations. For the normal distribution, the support refers to the $4\sigma$ interval ($>99.99\%$ of all values). For the gamma distribution, special cases occur when $a$ is an integer (Erlang distribution) and when $a=1$ (exponential distribution).}
    \label{tab:distributions}
    \vspace*{-0.5cm}
\end{table}
For the simulations below, we fix $m=210, w=20$ and for the gamma distribution, $a=0.5$, and allow $d$ to vary between $(0,0.5]$ as mentioned above. In all cases, we assume mean arrival rate $4$ times the target, $T$, i.e., two times the block size. The results are shown in \Cref{fig:systematic_d_lin,fig:systematic_d_exp}. The blue dots (trajectories of base fess and block sizes) show data from $Niter=100$ iterations after skipping $Nskip=200$ iterations. Allowing more iterations (either Nskip or Niter) does not produce different attractors which suggests that the plots show the steady state, i.e., what happens after the mechanism has converged. The initial base fee is set to $b_0=170$.\footnote{The Jupyter notebooks that have been used to generate these plots are available as supplementary material (or upon request).}

Qualitatively similar results are obtained for different values of parameters $m,w$ and $a$ (not presented here).

\paragraph{Range of valuations:} For this part, we fix $d$ to its default value, $d=0.125$, and vary $w$ which parameterizes the standard deviation (variance) of the distribution of users valuations. The range of users valuations has a significant effect on the stability of both BFURs since more (less) concentrated valuations tend to produce more (less) aggressive changes in block sizes. The rest of the parameters, i.e., Niter, Nskip, $b_0$, are the same as above. For all cases, we allow $w$ to vary between $0$ (very concentrated) and $20$ (less concentrated). The results are presented in \Cref{fig:systematic_w_lin,fig:systematic_w_exp}.

\clearpage
\begin{figure}[!th]
    \centering
    \includegraphics[width=0.91\linewidth]{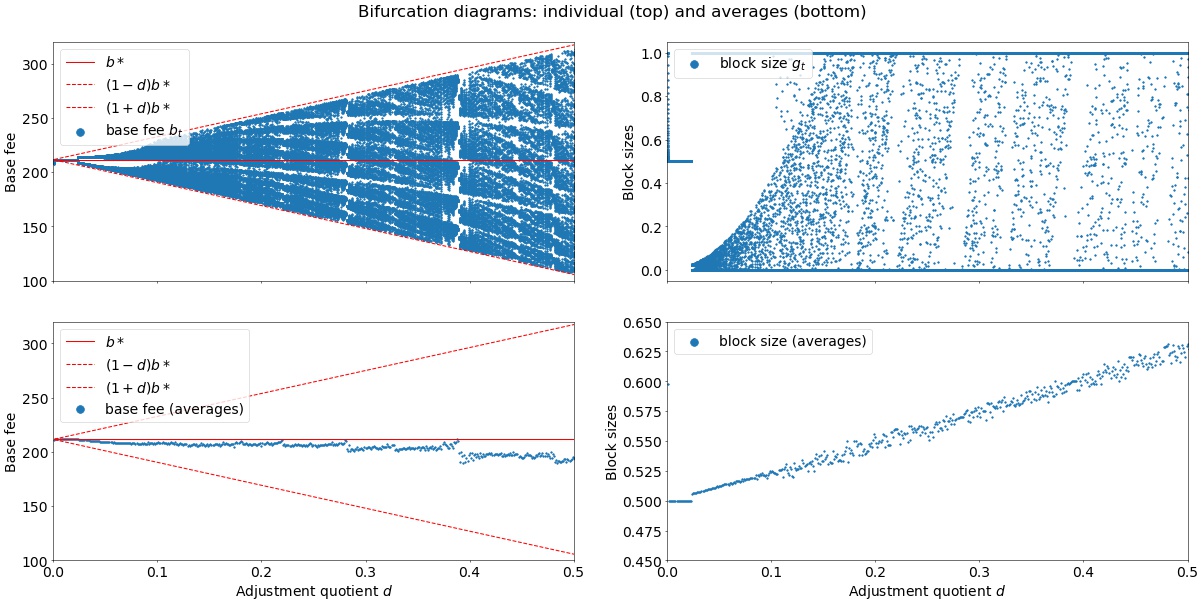}\\
    \includegraphics[width=0.91\linewidth]{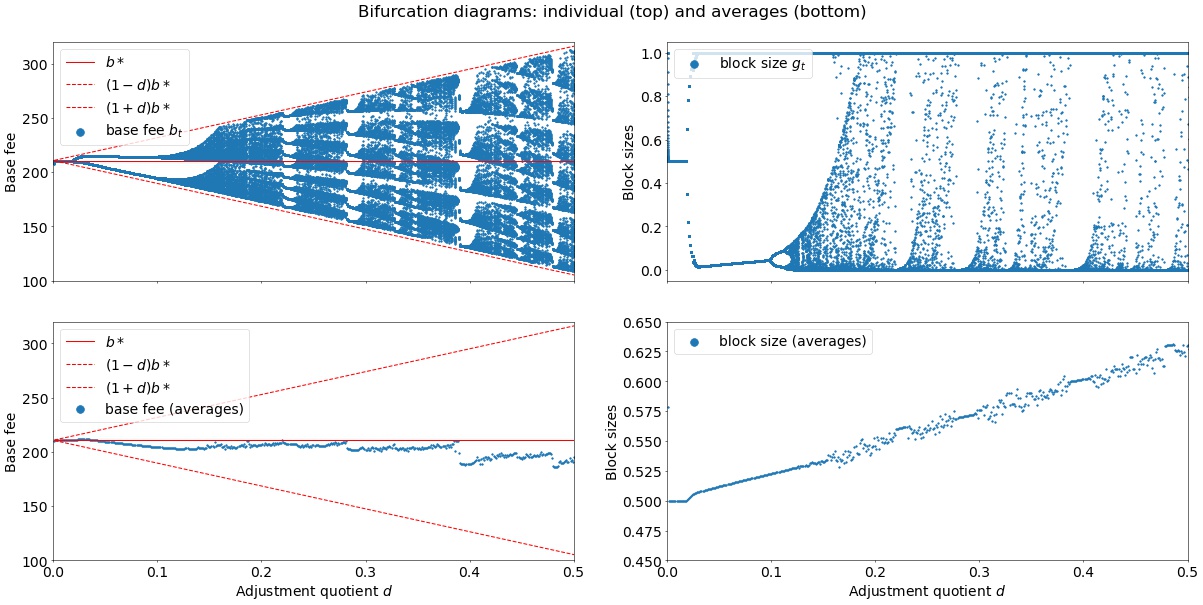}\\
    \includegraphics[width=0.91\linewidth]{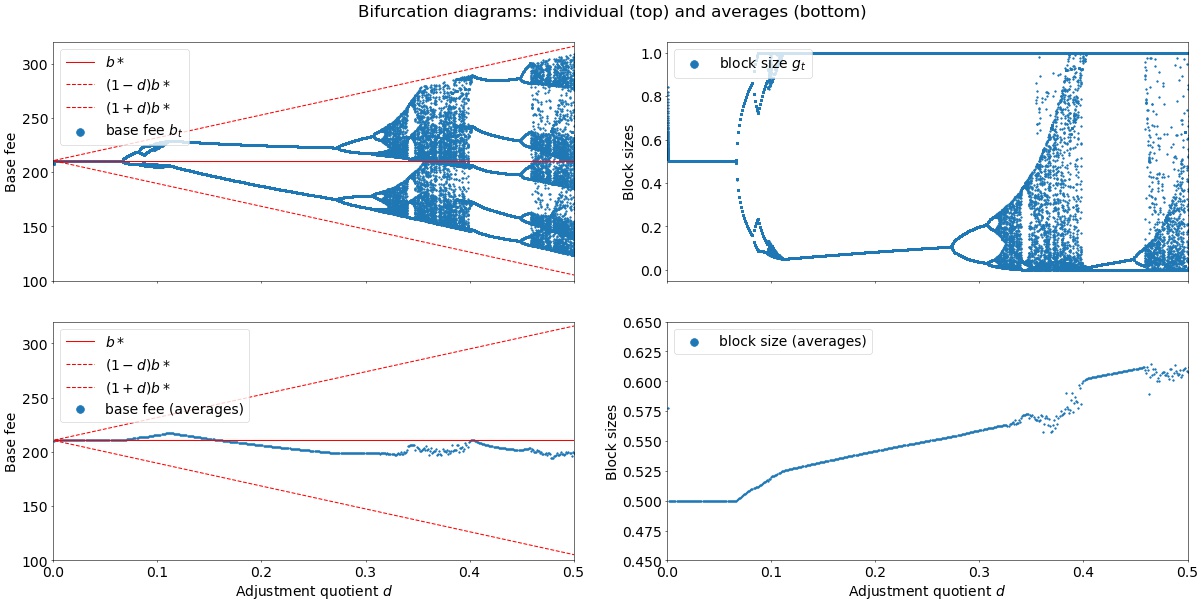}
    \caption{Simulations of \eqref{eq:eip1559} for different values of the adjustment quotient (bifurcation parameter, $d$) for uniform (top), normal (middle) and gamma (bottom) distributions. Note: similar to \Cref{fig:bifurcation_diagrams,fig:bifurcation_diagrams_exp}, the scale of the $y-$axis in block-size panels (top- and bottom-right) are between $0$ and $1$ (target block size is $0.5$).}
    \label{fig:systematic_d_lin}
\end{figure}

\begin{figure}[!th]
    \centering
    \includegraphics[width=0.91\linewidth]{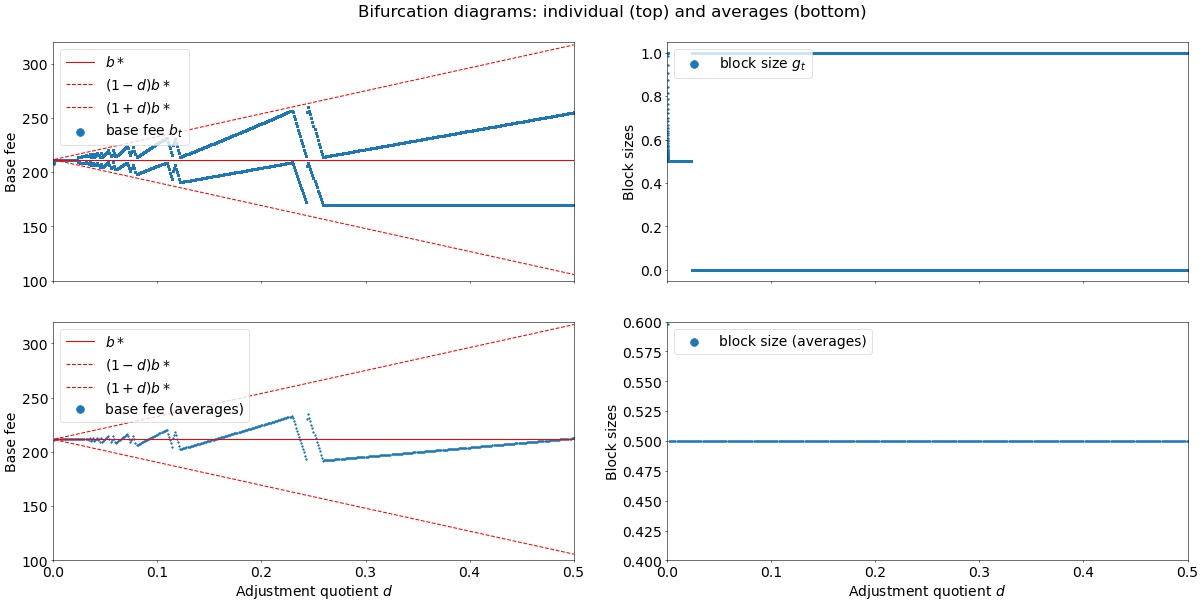}\\
    \includegraphics[width=0.91\linewidth]{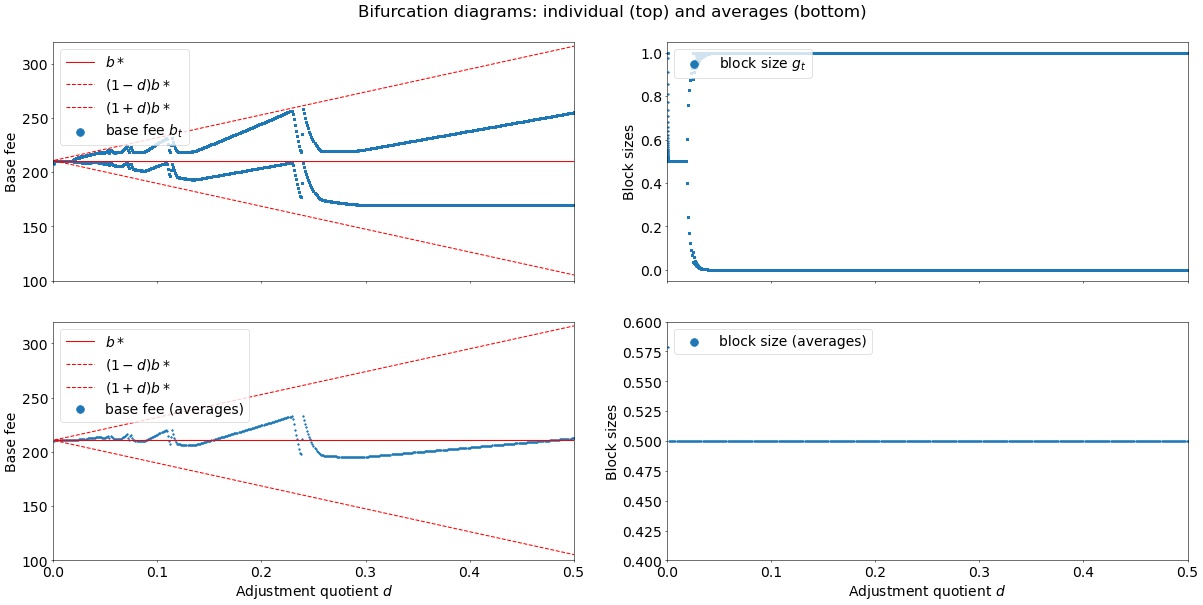}\\
    \includegraphics[width=0.91\linewidth]{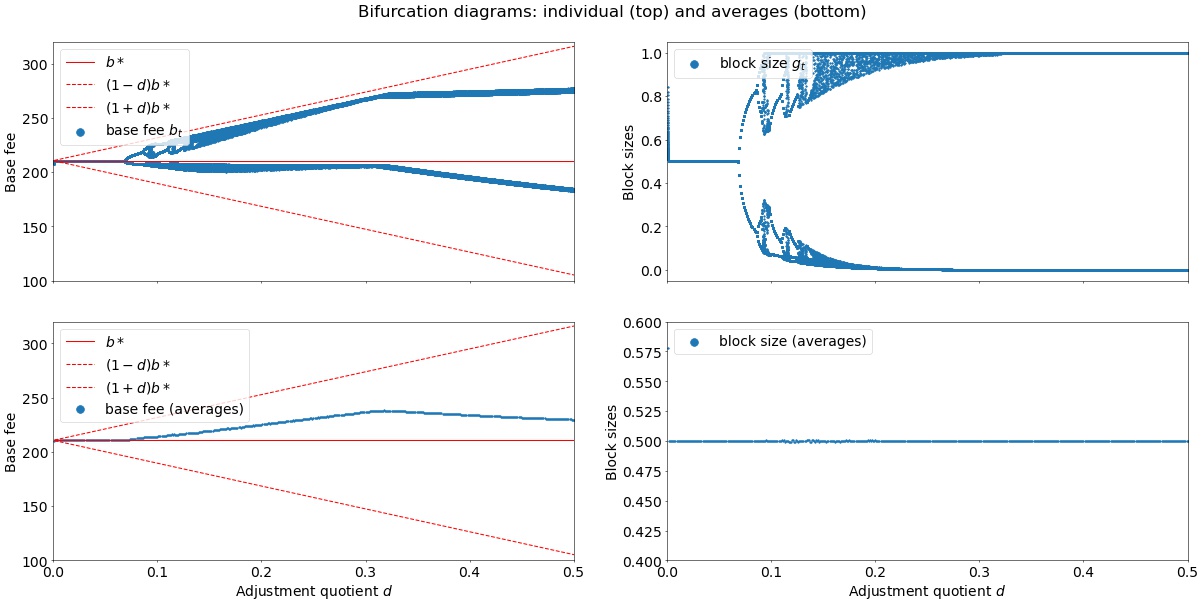}
    \caption{Simulations of \eqref{eq:exponential} for different values of the adjustment quotient (bifurcation parameter, $d$) for uniform (top), normal (middle) and gamma (bottom) distributions. The individual base fee trajectory depends on the initial conditions, but averages and block sizes are robust regardless of the exact base fee realization.}
    \label{fig:systematic_d_exp}
\end{figure}

\begin{figure}[!th]
    \centering
    \includegraphics[width=0.902\linewidth]{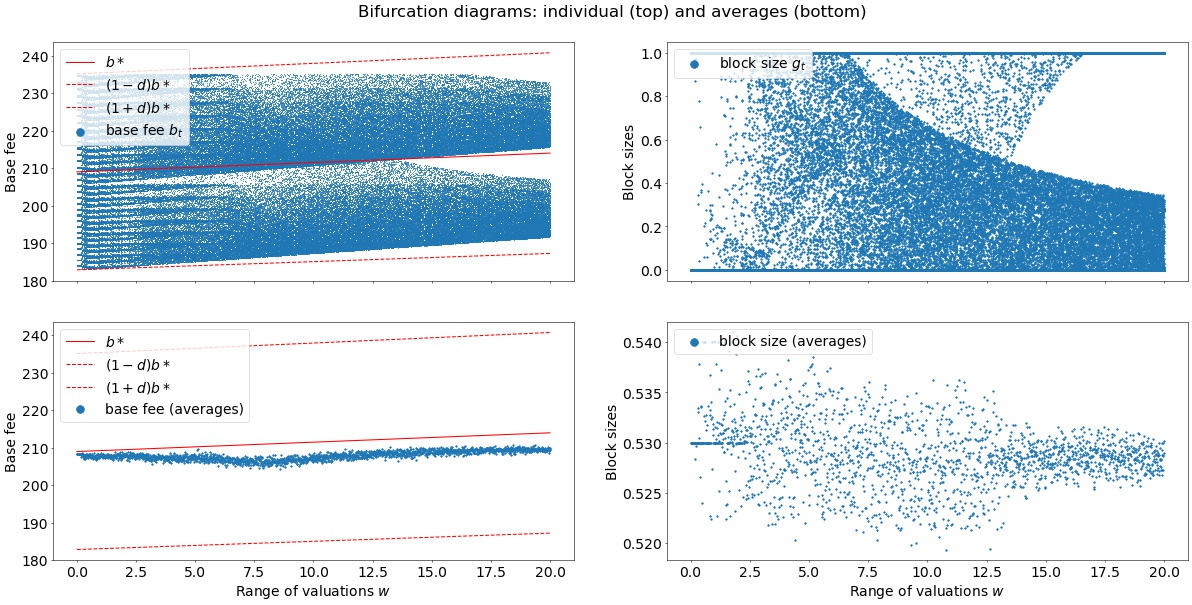}\\
    \includegraphics[width=0.902\linewidth]{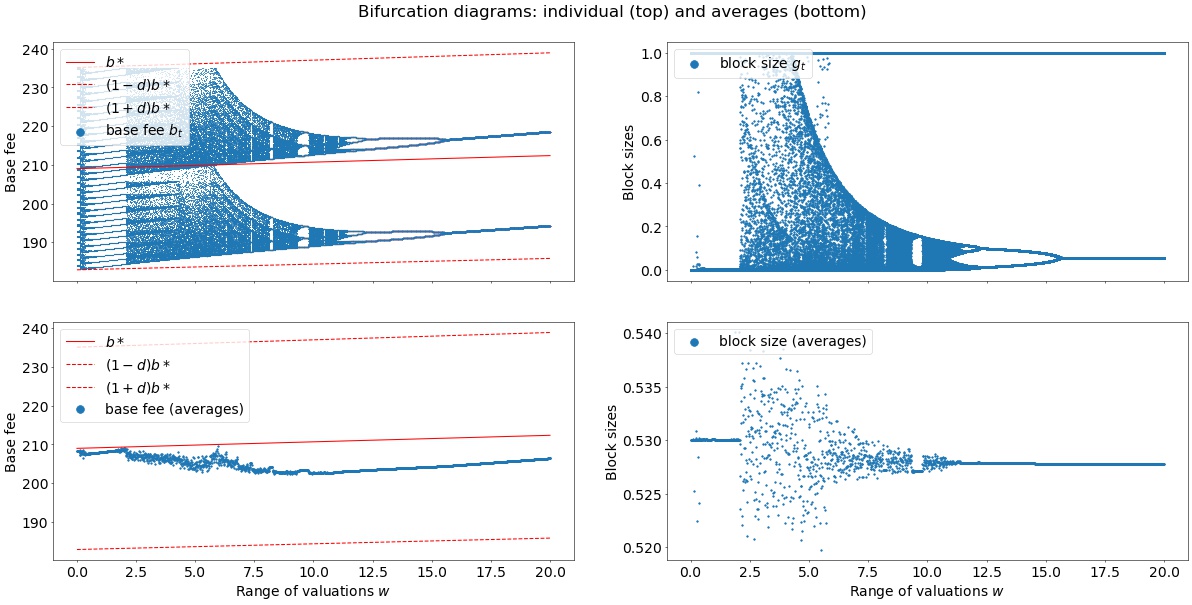}\\
    \includegraphics[width=0.902\linewidth]{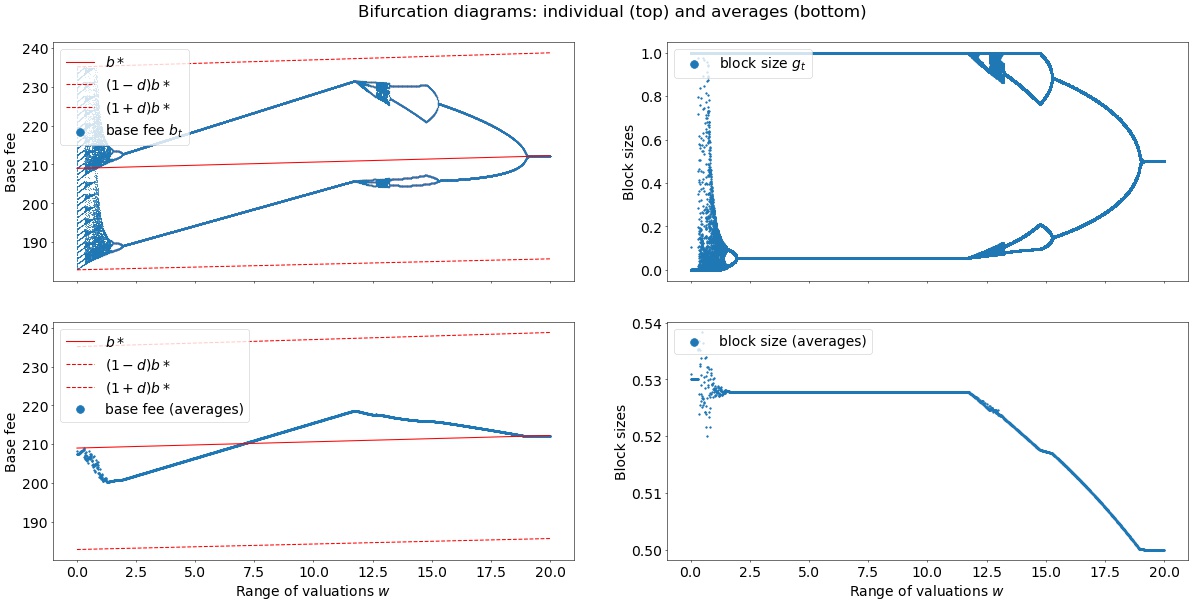}
    \caption{Simulations of \eqref{eq:eip1559} for default $d=0.125$ and uniform (top), normal (middle) and gamma (bottom) distributions with varying range of valuations (bifurcation parameter, $w$). The individual base fee trajectory depends on the initial conditions, but averages and block sizes are robust regardless of the exact base fee realization.}
    \label{fig:systematic_w_lin}
\end{figure}

\begin{figure}[!th]
    \centering
    \includegraphics[width=0.91\linewidth]{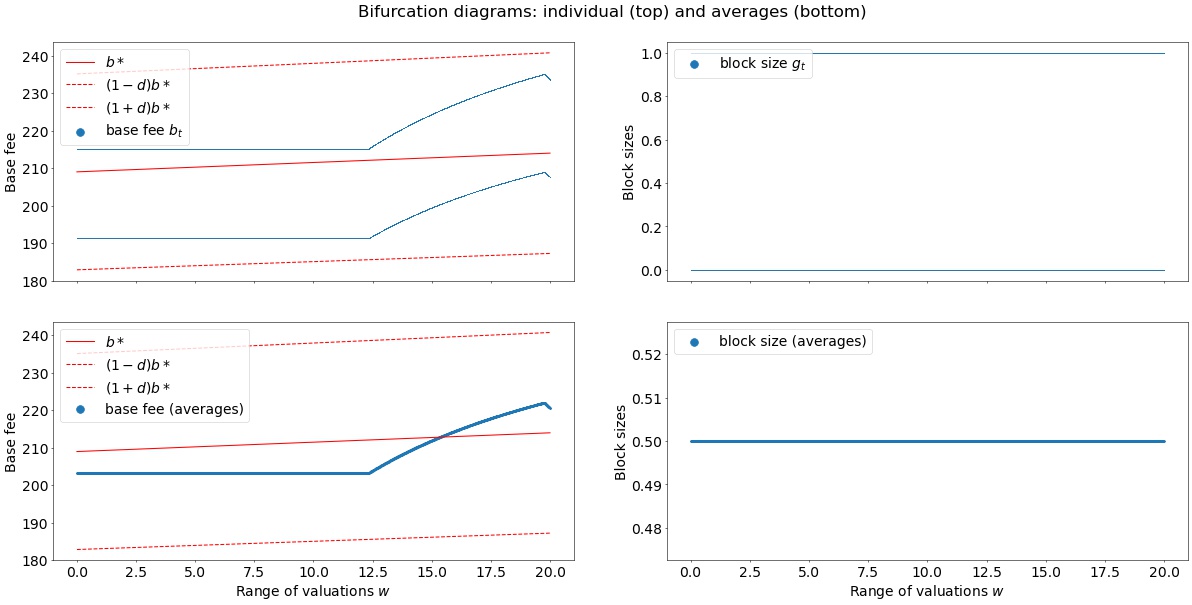}\\
    \includegraphics[width=0.91\linewidth]{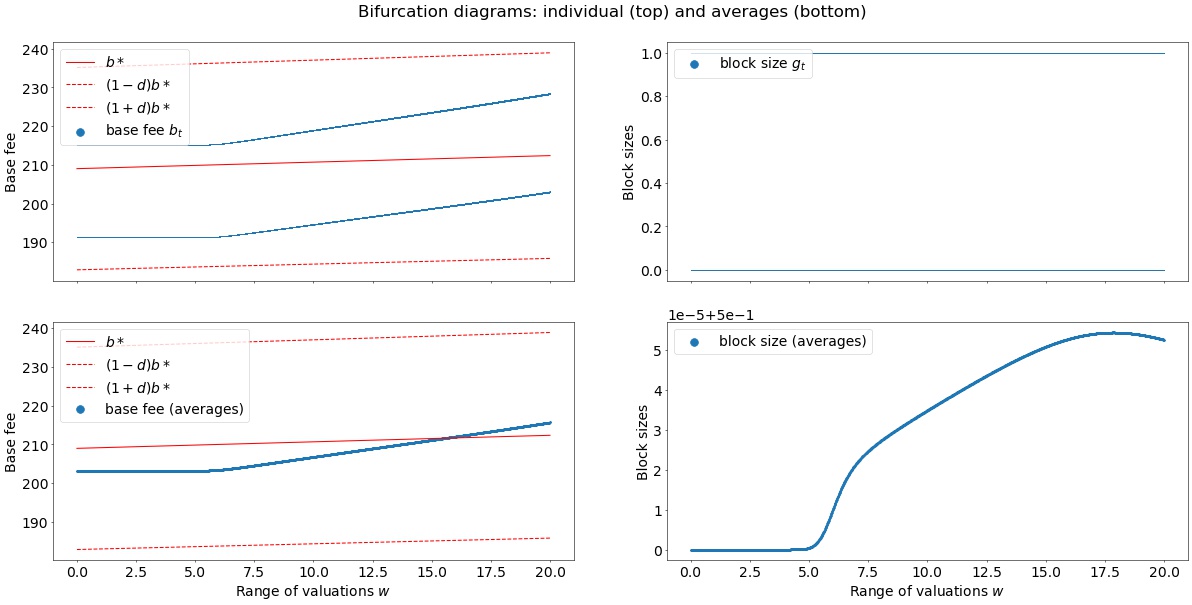}\\
    \includegraphics[width=0.91\linewidth]{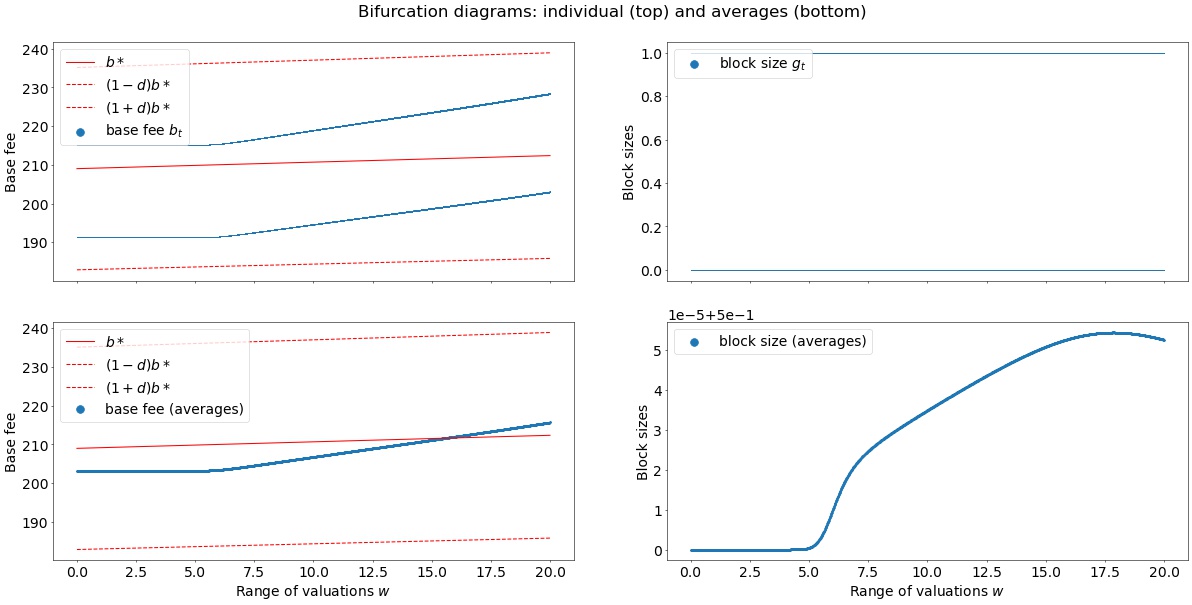}
    \caption{Simulations of \eqref{eq:exponential} for default $d=0.125$ and uniform (top), normal (middle) and gamma (bottom) distributions with varying range of valuations (bifurcation parameter, $w$). For the average block size plots (bottom left panels), note that $1e-5+5e-1 = 0.5001$.}
    \label{fig:systematic_w_exp}
\end{figure}

\end{document}